\documentclass[sigconf,screen]{acmart}

\copyrightyear{2024} 
\acmYear{2024} 
\setcopyright{rightsretained} 
\acmConference[LICS '24]{39th Annual ACM/IEEE Symposium on Logic in Computer Science}{July 8--11, 2024}{Tallinn, Estonia}
\acmBooktitle{39th Annual ACM/IEEE Symposium on Logic in Computer Science (LICS '24), July 8--11, 2024, Tallinn, Estonia}\acmDOI{10.1145/3661814.3662091}
\acmISBN{979-8-4007-0660-8/24/07}

\usepackage{tikz-cd}
\usepackage{RMC}
\tikzstyle{termbox}=[draw=term,fill=term!10,rounded corners,minimum size=20pt]
\tikzstyle{tallbox}=[draw=term,fill=term!10,rounded corners,minimum width=20pt,minimum height=40pt]
\tikzstyle{termpic}=[x=1pt,y=1pt,inner sep=0pt,outer sep=0pt,thick]

\usepackage{algpseudocode}
\usepackage{fancyvrb}
\usepackage{mathtools}
\usepackage{proof}
\usepackage{tikzit}

\tikzstyle{box}=[shape=rectangle, text height=1.5ex, text depth=0.25ex, yshift=0.5mm, fill=white, draw=black, minimum height=12.5mm, yshift=-0.5mm, minimum width=7.5mm, font={\small}]
\tikzstyle{Z dot}=[inner sep=0mm, minimum size=2mm, shape=circle, draw=black, fill={rgb,255: red,160; green,255; blue,160}]
\tikzstyle{Z phase dot}=[minimum size=5mm, font={\footnotesize\boldmath}, shape=rectangle, rounded corners=2mm, inner sep=0.2mm, outer sep=-2mm, scale=0.8, tikzit shape=circle, draw=black, fill={rgb,255: red,160; green,255; blue,160}, tikzit draw=blue]
\tikzstyle{X dot}=[Z dot, shape=circle, draw=black, fill={rgb,255: red,220; green,0; blue,0}]
\tikzstyle{X phase dot}=[Z phase dot, tikzit shape=circle, tikzit draw=blue, fill={rgb,255: red,220; green,0; blue,0}, font={\footnotesize\color{white}\boldmath}]
\tikzstyle{hadamard}=[fill=yellow, draw=black, shape=rectangle, inner sep=0.6mm, minimum height=1.5mm, minimum width=1.5mm]
\tikzstyle{small hadamard}=[hadamard]
\tikzstyle{vertex}=[inner sep=0mm, minimum size=1mm, shape=circle, draw=black, fill=black]
\tikzstyle{vertex set}=[inner sep=0mm, minimum size=1mm, shape=circle, draw=black, fill=white, font={\footnotesize\boldmath}]
\tikzstyle{new style 0}=[fill=none, draw=none, shape=rectangle, font={\large}]

\tikzstyle{simple}=[-]
\tikzstyle{hadamard edge}=[-, color=blue, dashed, dash pattern=on 2pt off 0.7pt]
\tikzstyle{brace edge}=[-, tikzit draw=blue, decorate, decoration={brace,amplitude=1mm,raise=-1mm}]
\tikzstyle{gray}=[-, draw={rgb,255: red,191; green,191; blue,191}]
\tikzstyle{arrow}=[<-, draw={rgb,255: red,128; green,128; blue,128}]
\tikzstyle{double-arrow}=[draw={rgb,255: red,128; green,128; blue,128}, <->]
\tikzstyle{new big text}=[-]

\newcommand\KA{\mathsf{k}}
\newcommand\embed[1]{\llbracket#1\rrbracket}

\AtBeginDocument{%
  }

\begin{document}

\title{The Relational Machine Calculus}

\author{Chris Barrett}
\email{chris.barrett@cs.ox.ac.uk}
\orcid{0000-0003-1708-3554}
\affiliation{
  \institution{University of Oxford}
  \city{}
  \country{United Kingdom}
}

\author{Daniel Castle}
\email{drc22@bath.ac.uk}
\orcid{0009-0000-3333-8161}
\affiliation{
  \institution{University of Bath}
  \city{}
  \country{United Kingdom}
}

\author{Willem Heijltjes}
\email{w.b.heijltjes@bath.ac.uk}
\orcid{0009-0001-8941-1150}
\affiliation{
  \institution{University of Bath}
  \city{}
  \country{United Kingdom}
}

\renewcommand{\shortauthors}{C.\ Barrett, D.\ Castle, and W.\ Heijltjes}

\begin{abstract}
This paper presents the Relational Machine Calculus (RMC): a simple,  foundational model of first-order relational programming.  The RMC originates from the Functional Machine Calculus (FMC),  which generalizes the lambda-calculus and its standard call-by-name stack machine in two directions.  One, "locations", introduces multiple stacks, which enable effect operators to be encoded into the abstraction and application constructs. The second, "sequencing", introduces the imperative notions of "skip" and "sequence", similar to kappa-calculus and concatenative programming languages.

The key observation of the RMC is that the first-order fragment of the FMC exhibits a latent duality which, given a simple decomposition of the relevant constructors,  can be concretely expressed as an involution on syntax.  Semantically,  this gives rise to a sound and complete calculus for string diagrams of Frobenius monoids.  

We consider unification as the corresponding symmetric generalization of beta-reduction. 
By further including standard operators of Kleene algebra,  the RMC embeds a range of computational models: the kappa-calculus,  logic programming,  automata, Interaction Nets,  and Petri Nets,  among others.  These embeddings preserve operational semantics,  which for the RMC is again given by a generalization of the standard stack machine for the lambda-calculus.  
The equational theory of the RMC (which supports reasoning about its operational semantics) is conservative over both the first-order lambda-calculus and Kleene algebra,  and can be oriented to give a confluent reduction relation.
\end{abstract}

\ccsdesc[500]{Theory of computation~Abstract machines}
\ccsdesc[500]{Theory of computation~Lambda calculus}
\ccsdesc[500]{Theory of computation~Regular languages}
\ccsdesc[500]{Theory of computation~Operational semantics}
\ccsdesc[500]{Theory of computation~Denotational semantics}
\ccsdesc[500]{Theory of computation~Categorical semantics}
\ccsdesc[500]{Theory of computation~Constraint and logic programming}

\keywords{lambda-calculus,  Kleene algebra,  logic programming,  reversible programming,  non-determinism,  hypergraph category,  Krivine abstract machine,  categorical semantics,  operational semantics.}


\maketitle


\section{Introduction}

The $\lambda$-calculus is widely considered the canonical model of functional programming.  However,  there is no similarly agreed foundation of \emph{relational} programming.  Paradigmatic examples of relational programming are logic programming,  based on unification \cite{Krishnaswami:Datafun,  Miller:HOLP, hanus_functional_2013, somogyi_execution_1996, Friedman:microkanren},  and database query languages,  based on Tarski's relational algebra \cite{Codd:relational-database,  tarski:1941}. 
Relations also serve as the (intended) semantics for a wide range of languages.  For example,  non-deterministic finite automata,  Kleene algebra and its variants --- some of which subsume propositional Hoare logic \cite{conway:regex, Kozen:KAT} --- and monadic languages for non-determinism \cite{Moggi-monad}. 
Furthermore,  the category of relations is a simple folk model of (differential) linear logic and thus of the linear $\lambda$-calculus; via the Kleisli construction,  it is also the basis of a simple quantitative model of the plain $\lambda$-calculus \cite{ehrhard_differential_2003,girard_normal_1988, laird_weighted_2013}, which can even be seen as underlying game semantics and intersection types \cite{Ghica-McCusker-2003, Carvalho:intersection,ong_quantitative_2017,calderon_understanding_2010}.   

Here we present a foundational model of (first-order,  sequential) relational programming: the \emph{Relational Machine Calculus (RMC)}.  
We first set out the aims and constraints that informed its design. Then, in the remainder of the introduction,  we detail the origin of the RMC: from a first-order $\lambda$-calculus through several adaptations to incorporate the relational paradigm,  including \emph{duality},  \emph{unification}, and \emph{non-determinism}. The body of the paper is dedicated to justifying our claim of meeting the design criteria we now set out.


\subsection{Design criteria}

The space of all possible (first-order,  sequential) relational models of computation is bewildering,  and indeed there is a proliferation of languages in this space,  indicating its importance  and range of applications.  Yet,  we find none in the literature which satisfy the following --- minimal,  but stringent --- design criteria that we would expect to be met by a "$\lambda$-calculus of relational computation". 
\begin{description}
	\item[Denotational semantics:]
has a relational and categorical semantics,  through quotienting by a local equational theory;

	\item[Duality:] exhibits a syntactic involution which switches {input} and {output},  and denotes the relational {converse}; 

	\item[Operational semantics:]
preserves that of standard models of first-order,  sequential (relational) computation;

	\item[Confluence:]
orienting the equational theory gives a confluent reduction relation,  sound for the operational semantics.
\end{description}
We proceed to motivate these criteria,  which serve to tame the design space and establish standards for success in our programme.  

%
\paragraph{Denotational semantics}
A minimal qualifying criterion for any relational calculus is to have a denotational semantics in the category of sets and relations.  We further expect a categorical semantics in an appropriate symmetric monoidal category which abstracts the structural properties of relations; we shall see momentarily our preferred candidate. 

%
\paragraph{Duality}
While the category of sets and relations embeds that of sets and functions, it exhibits very different structural properties; mainly,  a perfect duality between input and output given by relational converse.  We consider this {duality} a defining feature of relational computation.  A paradigmatic example is the following Prolog program to concatenate two lists. 
\begin{Verbatim}
  concat([],L,L)
  concat([E | L1], L2, [E | L3]) :- concat(L1,L2,L3)
\end{Verbatim}
Fixing input lists \textsf{L1} and \textsf{L2} in the relation \textsf{concat(L1,L2,L3)} returns their concatenation as \textsf{L3},  but by fixing the value \textsf{L3} instead,  the relation may be run in reverse to (non-deterministically) return every way of splitting \textsf{L3} into \textsf{L1} and \textsf{L2}.

Duality is a common theme in programming language research  \cite{SELINGER_2001,Filinski:duality,Girard:linear-logic},  and in mathematics and physics more generally.  Its aesthetic appeal barely needs justification; practically,  it offers parsimony of expression by identifying a dual program (theorem) with each program (theorem) written.  We require our language to feature duality in a direct and natural way: by a \emph{syntactic} involution.

%
\paragraph{Operational semantics}
We expect our language to be expressive enough to encode a range of relational models of first-order,  sequential computation, including the first-order $\lambda$-calculus,  automata, logic programming,   and Petri nets.  However,  any sufficiently expressive language can encode any other:  we would like the embeddings into our calculus to be, in some way, \emph{natural}. As our criterion for this, we specify that operational semantics is preserved,  for models where it is defined.  This means that (at least,  with respect to \emph{operational} semantics --- a relatively fine-grained notion) the encoded models may be viewed as \emph{fragments} of our language,  built from its more fundamental primitives.

%
\paragraph{Confluence}
The intended relational semantics will impose an equational theory on our calculus. Following the central notion of $\beta$-reduction in the $\lambda$-calculus, we expect the equational theory (modulo certain congruences) to be orientable as a reduction relation,  and that this relation is \emph{confluent}, contributing in an essential way to an effective solution to this theory.


\subsection{First-order lambda-calculus and duality}

Our point of departure is the first-order $\lambda$-calculus, as first appeared as the $\kappa$-calculus~\cite{Hasegawa:decomposing-LC,Power:kappa-cat} and recently as the first-order fragment of the Functional Machine Calculus~\cite{Heijltjes:FMCI,Heijltjes:FMCII} (FMC).  These calculi feature three of our desired properties: a \emph{denotational semantics} (albeit in Cartesian categories),  an \emph{operational semantics} in a simple stack machine,  and \emph{confluent reduction}. Crucially, they also exhibit the potential for \emph{duality}; and in the formulation of the FMC, this is directly observable in the syntax. 

Taking the operational perspective, the first-order FMC is an instruction language for a simple stack machine. Terms are sequences of \emph{push} $\term{[x]}$ and \emph{pop} $\term{<x>}$ operations over variables $\term x$ or constant values $\val c$, separated by sequential composition $(\term{;})$. The following example shows a term and its associated string diagram; the type is given by the size of its input and output stacks.
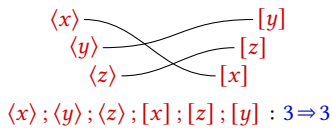
\[
\begin{array}{c}
	\begin{tikzpicture}
	\begin{pgfonlayer}{nodelayer}
		\node [style=none] (0) at (-2.5, 1.75) {};
		\node [style=none] (1) at (-1.5, 0.25) {};
		\node [style=none] (2) at (1.5, 1) {};
		\node [style=none] (3) at (1, 0.25) {};
		\node [style=none] (9) at (1.5, 0.25) {$\term{[x]}$};
		\node [style=none] (10) at (-2, 1) {};
		\node [style=none] (11) at (2, 1.75) {};
		\node [style=none] (12) at (2, 1) {$\term{[z]}$};
		\node [style=none] (13) at (-3, 1.75) {$\term{<x>}$};
		\node [style=none] (14) at (-2.5, 1) {$\term{<y>}$};
		\node [style=none] (15) at (-2, 0.25) {$\term{<z>}$};
		\node [style=none] (16) at (2.5, 1.75) {$\term{[y]}$};
	\end{pgfonlayer}
	\begin{pgfonlayer}{edgelayer}
		\draw [in=0, out=180] (3.center) to (0.center);
		\draw [in=180, out=0] (1.center) to (2.center);
		\draw [in=-180, out=0, looseness=1.25] (10.center) to (11.center);
	\end{pgfonlayer}
\end{tikzpicture}
\\
	\term{<x>;<y>;<z>;[x];[z];[y] ~: 3 > 3}
\end{array}
\]
Operationally, a term evaluates from left to right; $\term{<x>}$ pops the head off the stack, say a constant $\val c$, and substitutes $\val c$ for $\term x$ in the remaining computation; $\term{[x]}$ pushes the value substituted for $\term x$ onto the stack. We write our stacks with the head to the right, to match the order of \emph{pushes}; our term then takes a stack $\val{e\,d\,c}$ to $\val{c\,e\,d}$, where $\term{<x>}$ pops $\val c$, $\term{<y>}$ pops $\val d$, and $\term{<z>}$ pops $\val e$.

In this way, our term implements a relation between input and output stacks, as illustrated by the string diagram, where the wires represent the manipulation of the items on the stack (shown with the head at the top). The simple typing discipline of giving input and output arity tames the asymmetric nature of stacks, resulting in an internal language for symmetric monoidal categories (SMCs) in the linear case (where each \emph{pop} $\term{<x>}$ is matched by a unique \emph{push} $\term{[x]}$ to its right, and vice versa).

In the FMC, a \emph{pop} $\term{<x>;M}$, with $\term M$ the remaining computation, is a first-order lambda-abstraction, which we write $\kappa x.M$ following the $\kappa$-calculus. A \emph{push} $\term{[x];M}$ corresponds to an application $\term{M\,x}$, restricted to first order by forcing the argument to be a variable. The stack machine is then a simplified Krivine machine~\cite{krivine_call-by-name_2007}, replacing environments with substitution. This highlights a key aspect. A $\kappa$-abstraction is a \emph{binder}, making the variable $x$ \emph{local} to $\kappa x.M$. However, \emph{binding} and \emph{variable scope} are the main notions that will need to be re-visited in light of duality. A pop $\term{<x>}$ is therefore not a binder, and variables in this fragment of the RMC are \emph{global}. This does not prevent us from encoding first-order $\lambda$-calculus: in the first-order setting, terms are never duplicated, removing any issues with variable capture or $\alpha$-conversion. We may then simulate binding by stipulating that each \emph{pop} $\term{<x>}$ has a unique variable, which only occurs in \emph{pushes} to its right (this is \emph{Barendregt's convention}). Later, when duplication of terms returns, we will re-introduce local variable scope with a binding \emph{new variable} construct $\term{Ex.M}$,  justifying the core fragment of the RMC as a decomposition of abstraction into its two distinct roles: \emph{pop} and \emph{new}.

%
\paragraph{Duality}
Our formulation of first-order $\lambda$-calculus is tailored to reveal its latent duality, effected simply by reversing a term while switching \emph{push} and \emph{pop}. The dual of our previous example term:
\[
	\term{<y>;<z>;<x>;[z];[y];[x]~: 3 > 3}~.
\]
The linear fragment, where terms represent \emph{permutations}, is closed under duality. By contrast, the \emph{non-linear} first-order $\lambda$-calculus is characterised by \emph{duplication} and \emph{deletion}, embodied by the \emph{diagonal} and \emph{terminal} below left. These yield the dual terms below right, which we label \emph{matching} and \emph{arbitrary}, that feature multiple related \emph{pops}, or zero.
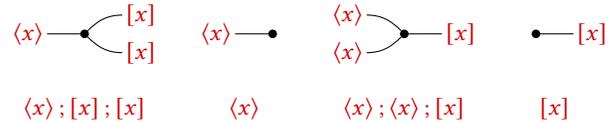
\[
	\begin{tikzpicture}
	\begin{pgfonlayer}{nodelayer}
		\node [style=none] (5) at (7.5, 1) {};
		\node [style=vertex] (6) at (6.5, 1) {};
		\node [style=none] (13) at (-1.5, 1) {};
		\node [style=vertex] (14) at (-0.5, 1) {};
		\node [style=none] (18) at (-1.25, -1) {$\term{<x>}$};
		\node [style=none] (20) at (7, -1) {$\term{[x]}$};
		\node [style=vertex] (25) at (-5.5, 1) {};
		\node [style=none] (26) at (-6.5, 1) {};
		\node [style=none] (27) at (-4.5, 1.5) {};
		\node [style=none] (28) at (-4.5, 0.5) {};
		\node [style=none] (29) at (-5.5, -1) {$\term{<x>;[x];[x]}$};
		\node [style=none] (30) at (-7, 1) {$\term{<x>}$};
		\node [style=none] (31) at (-4, 0.5) {$\term{[x]}$};
		\node [style=none] (32) at (-4, 1.5) {$\term{[x]}$};
		\node [style=none] (37) at (8, 1) {$\term{[x]}$};
		\node [style=none] (38) at (-2, 1) {$\term{<x>}$};
		\node [style=vertex] (39) at (3, 1) {};
		\node [style=none] (40) at (4, 1) {};
		\node [style=none] (41) at (2, 1.5) {};
		\node [style=none] (42) at (2, 0.5) {};
		\node [style=none] (43) at (3, -1) {$\term{<x>;<x>;[x]}$};
		\node [style=none] (44) at (1.5, 1.5) {$\term{<x>}$};
		\node [style=none] (45) at (1.5, 0.5) {$\term{<x>}$};
		\node [style=none] (46) at (4.5, 1) {$\term{[x]}$};
	\end{pgfonlayer}
	\begin{pgfonlayer}{edgelayer}
		\draw (6) to (5.center);
		\draw (14) to (13.center);
		\draw [in=45, out=180] (27.center) to (25);
		\draw [in=180, out=-45] (25) to (28.center);
		\draw (25) to (26.center);
		\draw [in=135, out=0] (41.center) to (39);
		\draw [in=0, out=-135] (39) to (42.center);
		\draw (39) to (40.center);
	\end{pgfonlayer}
\end{tikzpicture}
\]
Semantics tells us that the new terms should represent the relational converse of the diagonal and terminal functions, given respectively by the partial function that sends $(x,x)$ to $x$ and is otherwise undefined,  and the non-deterministic function sending the trivial input to all possible outputs.  

Operationally, the required behaviour of \emph{matching} $\term{<x>;<x>;[x]}$ for an input stack of two identical values, $\val c\,\val c$, is to return $\val c$,  and for distinct values, $\val d\,\val c$, is to fail. We achieve this by generalising \emph{pop} to hold also \emph{constants}, $\term{<c>}$, interpreted as an \emph{assertion} that the head of the stack is $\val c$,  with failure for any other constant. Then $\term{<x>;<x>;[x]}$ evaluates by first popping $\val c$, substituting it for $\val x$ in both \emph{pop} and \emph{push} to get $\term{<c>;[c]}$, which then pops and replaces the second value $\val c$ (or fails for $\val d$). 

The second term, \emph{arbitrary} $\term{[x]}$, illustrates that stack values must include variables, and hence must be substituted for: the term $\term{<c>}$ for a stack $\val x\,\val x$ incurs a substitution of $\val c$ for the second $\val x$. The behaviour of \emph{push} and \emph{pop} thus becomes perfectly symmetric, with substitutions in both the stack and the remaining term.

Altogether these relations are abstracted categorically as \emph{Frobenius monoids},  which we take to capture non-linear computation in the presence of duality,  given in string diagrams as above and satisfied by the \emph{Frobenius equation} below.
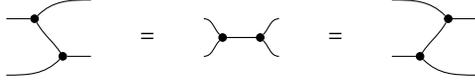
\[
	\begin{tikzpicture}
	\begin{pgfonlayer}{nodelayer}
		\node [style=none] (0) at (-6.25, -1) {};
		\node [style=vertex] (1) at (-5.5, 0.5) {};
		\node [style=none] (2) at (-6.25, 0.5) {};
		\node [style=none] (3) at (-4, 1) {};
		\node [style=none] (4) at (-4.75, -0.5) {};
		\node [style=vertex] (5) at (-4.75, -0.5) {};
		\node [style=none] (6) at (-4, -0.5) {};
		\node [style=vertex] (7) at (-0.5, 0) {};
		\node [style=none] (8) at (0.5, 0) {};
		\node [style=none] (9) at (-1, 0.5) {};
		\node [style=none] (10) at (-1, -0.5) {};
		\node [style=vertex] (11) at (0.5, 0) {};
		\node [style=none] (12) at (-0.5, 0) {};
		\node [style=none] (13) at (1, 0.5) {};
		\node [style=none] (14) at (1, -0.5) {};
		\node [style=none] (15) at (-2.5, 0) {$=$};
		\node [style=none] (16) at (6.25, -1) {};
		\node [style=vertex] (17) at (5.5, 0.5) {};
		\node [style=none] (18) at (6.25, 0.5) {};
		\node [style=none] (19) at (4, 1) {};
		\node [style=none] (20) at (4.75, -0.5) {};
		\node [style=vertex] (21) at (4.75, -0.5) {};
		\node [style=none] (22) at (4, -0.5) {};
		\node [style=none] (37) at (2.5, 0) {$=$};
	\end{pgfonlayer}
	\begin{pgfonlayer}{edgelayer}
		\draw [in=45, out=180] (3.center) to (1);
		\draw [in=120, out=-60] (1) to (4.center);
		\draw (1) to (2.center);
		\draw (5) to (6.center);
		\draw [in=135, out=0] (9.center) to (7);
		\draw [in=0, out=-135] (7) to (10.center);
		\draw (7) to (8.center);
		\draw [in=45, out=180] (13.center) to (11);
		\draw [in=180, out=-45] (11) to (14.center);
		\draw (11) to (12.center);
		\draw [in=0, out=-135] (5) to (0.center);
		\draw [in=135, out=0] (19.center) to (17);
		\draw [in=60, out=-120] (17) to (20.center);
		\draw (17) to (18.center);
		\draw (21) to (22.center);
		\draw [in=180, out=-45] (21) to (16.center);
	\end{pgfonlayer}
\end{tikzpicture}
\]
Intuitively, this implements the idea that \emph{connectivity = identity}: connected wires represent a single value. In our calculus, this is modeled by \emph{variable names}; and indeed the three diagrams in the above equation may all be represented by the term $\term{<x>;<x>;[x];[x]}$.

Frobenius monoids have received a huge amount of interest over the last decade or so,  since they were identified as a fundamental primitive of quantum computation by the ZX-calculus  \cite{coecke_interacting_2011, coecke_new_2013, Perdrix:completeness-ZX}.  Since then,  they have been adopted as primitives in a wide range of string diagrammatic languages,  including those capturing conjunctive queries,  relational algebra,  and aspects of logic programming 
\cite{Bonchi:graphical-conjunctive-queries,bonchi:tape-diagrams, zanasi:logic-programming,  bonchi_graphical_2019}.  In programming language theory specifically,  they have received significantly less attention,  but have been used in a synthetic axiomatization of the operation of \emph{exact conditioning} in probabilistic programming \cite{di_lavore_evidential_2023, Staton:exact-conditioning} and to model the (partial) inverse of duplication in reversible languages \cite{kaarsgaard_join_2021}.  The RMC offers a novel \emph{operational} account of Frobenius monoids,  thus filling a gap in the literature and bridging the study of string-diagrammatic and programming languages.


\subsection{The full calculus}

The calculus thus far is generated by: \emph{push} $\term{[x]}$ and \emph{pop} $\term{<x>}$ operations, \emph{sequential composition} $\term{M;N}$, and its unit $\term{*}$, the imperative \emph{skip}. Associativity of composition will be implemented via the machine and the equational theory. To reach our desired expressivity, we extend this with a careful selection of features, outlined below.

%
\paragraph{Algebraic terms and unification}
The $\beta$-reduction relation is the interaction of a consecutive \emph{push} and \emph{pop}, say $\term{[y];<x>}$, to incur a (global) substitution $\term{\{y/x\}}$. With constants, both $\term{[c];<x>}$ and $\term{[x];<c>}$ incur $\term{\{c/x\}}$, while $\term{[c];<c>}$ succeeds $(\term *)$ and $\term{[d];<c>}$ fails. 

Viewing a redex as a formal equation, these are the familiar rules of \emph{first-order unification},  restricted to variables and constants. We generalize our calculus to allow \emph{algebraic terms}\footnote{We write "algebraic terms" to avoid name clashing with "first-order terms",  which here refer to terms of the RMC,  in contrast with "higher-order terms" of the $\lambda$-calculus.} as values, and will implement unification on the machine and through a symmetric version of $\beta$-reduction.  
Evaluating a redex $\term{[t];<s>}$ will thus produce substitutions constituting the most general unifier of $\val t$ and $\val s$, or fail if none exists.

%
\paragraph{Kleene Algebra}
To internalize the partial and non-deterministic semantics of relations,  we introduce non-deterministic \emph{sum} $(\term+)$ and \emph{failure} $(\term{0})$,  its unit.  Indeed,  the presence of non-determinism is often considered a defining feature of relational programming. With sequencing already present,  we also desire a Kleene star construction $(\term{-^*})$ to model infinitary behaviour,  \textit{e.g.},  of logic programming.  Our full calculus will thus be conservative over Kleene algebra (KA).  

The atoms of KA may be interpreted as programs,  with the axioms of KA acting as a set of (weak) fundamental laws any standard non-deterministic,  sequential language should satisfy~\cite{hoare_complete_1987,Hoare:laws-of-programming,Kozen:KAT}.  However,  an \emph{operational} interpretation of KA is typically only given via its correspondence with finite state automata.  Our calculus can be viewed as a modification of KA to include the dual primitives of \emph{push} and \emph{pop},   replacing the atoms of KA and supporting its extension with a \emph{direct} operational semantics.

%
\paragraph{Local variables}
Evaluation for Kleene star, taking a term $\term{M^*}$ non-deterministically to a sequence $\term{M;\dots;M}$ of any length, re-introduces duplication of terms into the calculus, bringing with it the familiar problems of variable identity and $\alpha$-conversion.  We adopt a standard solution: making the scope of a variable explicit with a \emph{new variable} construct,  $\term{Ex.M}$,  familiar variously from nominal Kleene algebra \cite{Gabbay:nominal-KA},  functional logic programming \cite{hanus_functional_2013},  $\pi$-calculus \cite{Milner:pi-calculus},  and the $\nu$-calculus \cite{Pitts:nu-calculus} as a fragment of ML.  This restores locality and gives a proper decomposition of $\kappa$-abstraction into its two roles: first, as the binder of new variables,  and second, as an instruction to \emph{pop} from the stack.
\[
	\term{\kappa x.M} ~=~ \term{Ex.(<x>;M)}
\]

%
\paragraph{Locations} 
We inherit a further feature from the Functional Machine Calculus: multiple stacks (or streams) on the abstract machine, indexed in a set of \emph{locations}. This captures several computational effects: mutable store, as stacks restricted to depth one; input and output, as dedicated streams; and a probabilistic generator, as a stream of random bits --- remarkably,  \emph{while retaining confluent reduction}~\cite{Heijltjes:FMCI}. The required generalization of the syntax is a simple parameterization of \emph{push} and \emph{pop} in a location $ a$, as $\term{[x]a}$ and $\term{a<x>}$ respectively, to operate on the indicated stack.  Thereby,  the RMC subsumes the operational semantics of the given effects, while $\beta\eta$-equivalence captures their algebraic theory~\cite{Heijltjes:FMCI, Heijltjes:FMCII}.  This extension allows the simple modelling of stateful models of computation such as Turing machines and Petri nets.

\section{The Relational Machine Calculus}
\label{sec:RMC}

We define the \emph{Relational Machine Calculus} (RMC) as follows. We assume a countable set of \emph{locations} ${A}=\{a,b,c\dots\}$, which each represent a stack on the abstract machine. Stacks hold algebraic terms over a signature $\Sigma$ of function symbols $f^n$ of arity $n$.

\begin{definition}[Relational Machine Calculus]
\emph{Values} $\val{s,t}$ and \emph{(computation) terms} $\term{M,N}$ are given by the grammars below.
\[
\begin{array}{@{}rrll@{}}
	\val{s,t} &\Coloneqq&
			\val{x}
	~\mid~	\val{f^n(t_1,\dots,t_n)} \quad (\val{f^n}\in\Sigma) & \text{algebraic terms}
\\ \\
	\term{M,N} &\Coloneqq&
			\term{*}
	~\mid~	\term{M;N}
	~\mid~	\term{M^*}
	~\mid~	\term{0}
	~\mid~	\term{M + N} & \text{Kleene algebra}
\\[5pt]
	&\mid&  \term{[t]a}
	~\mid~	\term{a<t>}  & \text{stack operations}
\\[5pt]
	&\mid&	\term{Ex.M}	 & \text{variable scope}
\end{array}
\]
\end{definition}

From left to right, the computation terms are \emph{skip} or \emph{nil} $\term*$, \emph{(sequential) composition}  $\term{M;N}$, \emph{Kleene star} $\term{M^*}$, \emph{zero} or \emph{failure} $\term 0$, a \emph{sum} of terms $\term{M+N}$, a \emph{push} $\term{[t]a}$ of the value $\val t$ to the location $a$, a \emph{pop} $\term{a<t>}$ from the location $a$ to unify with $\val t$, and a \emph{new} variable introduction $\term{Ex.M}$ which binds $x$ in $\term M$. Operator precedence: Kleene star binds tightest, then sequencing, then \emph{new} $\term{Ex.M}$, and finally sum; then $\term{Ex.M;N + P;Q^*}=\term{(Ex.(M;N)) + (P;(Q^*))}$.  We often omit the superscript on $f^n$. 

The calculus exhibits \emph{duality} as a syntactic involution.
  
\begin{definition}[Duality]
Define \emph{duality} $(-)^\dagger$ on terms as follows. 
\[
\begin{array}{rcl}
	      \term{*}^\dagger &=& \term{*} 
\\	  (\term{N;M})^\dagger &=& \term{M}^\dagger\term{;}\term{N}^\dagger
\\	  (\term{M^*})^\dagger &=& \term{(M}^\dagger\term{)^*}
\\	      \term{0}^\dagger &=& \term{0}
\end{array}
\qquad
\begin{array}{rcl}
	(\term{M + N})^\dagger &=& \term{M}^\dagger \, \term{+} \, \term{N}^\dagger
\\	    \term{[t]a}^\dagger &=& \term{a<t>}
\\	    \term{a<t>}^\dagger &=& \term{[t]a}
\\	 (\term{Ex.M})^\dagger &=& \term{Ex.M}^\dagger
\end{array}
\]
\end{definition}

We shall see throughout the paper how various operational,  equational,  and semantic notions either dualize or respect duality.


\subsection{Operational Semantics} 

The small-step operational semantics is given by a stack machine, the \emph{relational abstract machine} or \emph{relational machine}.

\begin{definition}[Relational Machine]
A \emph{state} is a triple $({S_A},\term M,K)$ of: a \emph{memory} $S_A$, a family of stacks indexed in a set of locations $A$; a term $\term M$; and a \emph{continuation stack} $K$. These are defined as follows.
\[
\begin{array}{rrll}
		K,L & \Coloneqq & \e ~\mid~ \term M\,K  & \text{continuation stacks}
\\[5pt]	S,T & \Coloneqq & \e ~\mid~ S\,\val t   & \text{operand stacks}
\\[5pt] S_A & \Coloneqq & \{ S_a \mid a\in A \} & \text{memories}
\end{array}
\]
The addition of a stack $S_a$ to a memory $S_A$ ($a\notin A$) is written $S_A\cdot S_a$. We abbreviate $\val{f(t_1,\dots,t_n)}$ to $\val{f(T)}$ where $T=\val{t_1\dots t_n}$, and $\term{[t_1]a;\dots;[t_n]a}$ to $\term{[T]a}$ and $\term{a<t_n>;\dots;a<t_1>}$ to $\term{a<T>}$ (note the inversion). The \emph{transitions} of the machine are given in Figure~\ref{fig:machine}, read top--to--bottom, and are non-deterministic: a state transitions to a formal sum of states, represented by branching of the transitions. A \emph{run} of the machine is a single rooted path in the machine tree, not necessarily to a leaf, shown with a double line as below. 
\[
	\run {S_A}MK {T_A}NL
\]
A run is \emph{successful} if it terminates in a state $\mstate{T_A}*\e$. A state where no transition rules apply represents \emph{failure}, and is considered to have \emph{zero} branches.
\end{definition}

\begin{figure}
\[
\begin{array}{@{}c@{}}
\begin{array}{@{}cc@{}}
		\step {S_A}*{\term M\,K}{S_A}MK
&		\begin{array}{@{}c@{}} \mstate{S_A}{M^*}K \\\hline \mstate{S_A}*K \quad \mstate{S_A}{M;M^*}K \end{array}
\\ \\	\step {S_A}{M;N}K {S_A}M{\term N\,K}
&		\begin{array}{@{}c@{}} \mstate{S_A}{M+N}K \\\hline \mstate{S_A}MK \quad \mstate{S_A}NK \end{array}
\\ \\	\step {S_A\cdot S_a}{[t]a}K {S_A\cdot S_a\,\val t}*K
&		\step {S_A}{Ex.M}K {S_A}{\{y/x\}M}K~(y~\text{fresh})
\\ \\	\step {S_A\cdot S_a\,\val x}{a<x>}K {S_A\cdot S_a}*K
&		\step {S_A\cdot S_a\,\val{f(R)}}{a<f(T)>}K {S_A\cdot S_a\,R}{a<T>}K
\end{array}
\\ \\	\step {S_A\cdot S_a\,\val x}{a<t>}K {\{\val t/\val x\}(S_A\cdot S_a)}*{\{\val t/\val x\}K}~(\val x\notin\val t)
\\ \\	\step {S_A\cdot S_a\,\val t}{a<x>}K {\{\val t/\val x\}(S_A\cdot S_a)}*{\{\val t/\val x\}K}~(\val x\notin\val t)
\end{array}
\]
\caption{Transitions of the Relational Machine}
\label{fig:machine}
\Description[Transitions of the Relational Machine]{}
\end{figure}

Machine evaluation of a term $\term M$ for an input memory $S_A$ gives a (possibly infinite) number of successful runs, each with a return memory $T_A$. The big-step evaluation function $(\eval{S_A}M)$ will collect these as a multiset. Since variables are global and machine steps substitute into the continuation stack, evaluation returns also a finite substitution map $\sigma$ with each return memory $T_A$, as a pair $(T_A,\sigma)$. In a composition $\term{M;N}$, the subsitutions from $\term M$ can then be applied to $\term N$.

\paragraph{Notation}
We denote the empty multiset by $\mempty$, a singleton by $\msingle{(T_A,\sigma)}$, and multiset union by $\munion$. Substitution maps $\sigma$ and $\tau$ are applied to a term as $\term{\sigma M}$ and composed as $\term{(\sigma\tau)M}=\term{\sigma(\tau M)}$. The empty map is $\e$, and $\restrict\sigma \val y$ is as $\sigma$ except undefined on $\val y$. The $n$-fold composition of a term $\term M$ is $\term{M^n}$ where $\term{M^0}=\term*$ and $\term{M^{n+1}}=\term{M;M^n}$.

\newcommand\floor[1]{\lfloor#1\rfloor}
\newcommand\ret{\mathsf{ret}}
\newcommand\sub{\mathsf{sub}}
\begin{definition}[Big-step operational semantics]
\label{def:big-step}
The big-step operational semantics of the RMC is given by the evaluation function $(\eval--)$ below, where $\val y$ is globally \emph{fresh} in the $\term{Ex.M}$ case.
\[
\begin{array}{@{}r@{\,\Downarrow\,}l@{~=~}l@{\qquad}r@{}}
	{S_A}                  & \term{*}    & \msingle{(S_A,\e)}
\\	{S_A}                  & \term{M;N}  & 
\multicolumn{2}{@{}l@{}}{
\begin{array}[t]{@{}l@{}l@{}}
	[\,(U_A,\tau\sigma) \mid \, & (T_A,\sigma) \in \eval{S_A}M,
    \\                          & (U_A,\tau)   \in \eval{T_A}{\sigma N}\,]
\end{array}}
\\	{S_A}                  & \term{M^*}  & \bigsqcup_{n\in\mathbb N}\,(\eval{S_A}{M^n})
\\	{S_A}                  & \term{M+N}  & (\eval{S_A}M) \munion (\eval{S_A}N)
\\	{S_A\cdot S_a}         & \term{[t]a} & \msingle{(S_A\cdot S_a\,\val t,\e)}
\\	{S_A\cdot S_a\,\val x} & \term{a<x>} & \msingle{(S_A\cdot S_a,\e)}
\\	{S_A\cdot S_a\,\val x} & \term{a<t>} & \msingle{(\{\val t/\val x\}(S_A\cdot S_a),\{\val t/\val x\})} & (\val x\notin\val t)
\\	{S_A\cdot S_a\,\val t} & \term{a<x>} & \msingle{(\{\val t/\val x\}(S_A\cdot S_a),\{\val t/\val x\})} & (\val x\notin\val t)
\\	{S_A\cdot S_a\,\val{f(R)}} & \term{a<f(T)>} & \eval{S_A\cdot S_a\,R}{a<T>}                  
\\	{S_A}                  & \term{Ex.M} & \multicolumn{2}{@{}l@{}}{\mcomp{(T_A,\restrict\sigma \val y)}{(T_A,\sigma)\in\eval{S_A}{\{y/x\}M}}}   
\\  {S_A}                  & \term M     & \mempty & \text{(otherwise)}
\end{array}
\]
The function $\floor-$ takes a multiset to its underlying set, and the function $\ret(-)$ 
projects onto only \emph{return memories} 
in 
$\floor{\eval{S_A}M}$.
\end{definition}


We assume that fresh variables are generated globally, without clashing. In particular, in $\term{(Ex.M);N}$ the variable instantiating the $\term{Ex}$ must be free for $\term N$ as well as $\term M$. One may address this practically by carrying along the \emph{formal parameters} of $\term{(Ex.M);N}$ into $\term{Ex.M}$; we omit this for brevity.

\begin{proposition}[Big-step semantics is well defined]\label{lem:big-step-well-defined}
The evaluation function $(\eval{-}{-})$ is a total function.
\end{proposition}


\begin{proposition}[Small-step and big-step semantics agree]\label{prop:small-big-agree}
For every memory $S_A$,  term $\term M$ and continuation $K$, there is a bijection between the elements $(T_A,\sigma)$ of $(\eval{S_A}M)$ and successful runs
\[
	\run{S_A}MK {T_A}*{\sigma K}~.
\]
\end{proposition}
%



\section{Encoding computational models, I}\label{sec:encoding-models-i}

We embed a number of computational models (or algorithms, in the case of unification) into the RMC to illustrate the origins and purposes of its various constructs.  Observe,  in particular,  that operational semantics is preserved,  often in  a strong sense.  We use only one, unnamed stack,  writing $\term{[t]}$ and $\term{<t>}$.  Encodings of stateful languages,  making use of multiple locations,  are given in Section \ref{sec:encoding-models-ii}.  

\paragraph{Notation}
We introduce vector notation for stacks of variables $\val{!x} = \val{x_1 \ldots x_n}$,  and emphasize the reversal in $\term{<?x>} = \term{<x_n>;\ldots;<x_1>}$ by pointing the arrow left.  Concatenation is given by juxtaposition.  We further abbreviate $\term{Ex_1\ldots Ex_n. M}$ as $\term{\exists !x.M}$.    



\subsection{Regular expressions}\label{subsec:regex}

Regular expressions (REs) over an alphabet $\Sigma$ of constants $a,b,c,\dots$ are given by the following grammar, with their embedding $\llbracket - \rrbracket$ into RMC-terms below it, given by juxtaposition.
\[
\begin{array}{r@{~~\Coloneqq~~}c@{~~\mid~~}c@{~~\mid~~}c@{~~\mid~~}c@{~~\mid~~}c@{~~\mid~~}c}
	          E,E' &   \e   &      E\,E' &       E^*  & \varnothing &     (E|E') & ~c\in\Sigma
\\[5pt] \term{M,N} & ~\term*~ &~ \term{M;N}~ & ~\term{M^*}~ &  ~ \term{0}~  & ~\term{M+N}~ & ~\term{[c]}
\end{array}
\]
%
Here, constant values $\val c$ are \emph{pushed} to the stack.  The regular language $\mathcal{L}_E$ defined by $E$ --- a set of words (stacks) over $\Sigma$ --- is thus given directly by the evaluation of $\llbracket E\rrbracket$. 

\begin{proposition}[Regular expressions embed]
For an RE $E$,
\[
	\mathcal{L}_E~=~\ret\floor{\evalnocolor{\e}{\llbracket E\rrbracket}}~.
\]
\end{proposition}

There is a \emph{dual} embedding $\llbracket - \rrbracket^\dagger$ of REs into the grammar below,  given by composition of the original embedding with duality. 
\[
	\term{M,N}~\Coloneqq~
		  \term*
	~\mid~\term{M;N}
	~\mid~\term{M^*}
	~\mid~\term{0}
	~\mid~\term{M+N}
	~\mid~\term{<c>}
\]
Here, constant values $\val c$ are \emph{popped} from the stack. The embedding can be considered as defining a computation which \emph{tests} words (stacks),  with a word $S$ accepted if $\e \in \ret\floor{\evalnocolor{S}{\llbracket E\rrbracket^\dagger}}$. The {characteristic (or indicator) function} of a regular language $\mathcal{L}_E$ is thus given directly by evaluation of $\llbracket E \rrbracket ^\dagger$.


\subsection{Unification}

A (first-order) unification algorithm~\cite{martelli_efficient_1982} takes a set of formal equations $E=\{{s_1}\doteq {t_1},\dots,{s_n}\doteq {t_n}\}$ and returns a most general unifier (MGU), a minimal subsitution $\sigma$ such that $\sigma {s_i}=\sigma {t_i}$ for all $i$, if one exists. In the RMC, we may encode such equations as redexes: $\llbracket s\doteq t\rrbracket=\term{[s];<t>}$.  Unification (as a general concept) can then be viewed as embedding in the RMC as the following fragment.
\[
	\term{M,N}~\Coloneqq~
	      \term*
	~\mid~\term{M;N}
	~\mid~\term{[t]}
	~\mid~\term{<t>}
\]
Terms are sequences of pushes and pops of first-order values to and from this location. A set of equations $E$ is encoded as follows.
\[
	\llbracket\{s_1\doteq t_1,\dots,s_n\doteq t_n\}\rrbracket
	=
	\term{[s_n];\dots;[s_1];<t_1>;\dots;<t_n>}
\]
Note that we have chosen nested redexes, since that is also what the machine reduction for function symbols produces; but a sequence of redexes would have been equally valid. We then observe that evaluation returns an MGU if one exists.

\begin{proposition}[First-order unification embeds]\label{prop:unification-embeds}
For a set of formal equations over algebraic terms $E$ we have the following, where $\sigma$ is an MGU for $E$, if one exists, and otherwise $\evalnocolor{\e}{\llbracket E\rrbracket}=\mempty$.
\[
	\evalnocolor{\e}{\llbracket E\rrbracket}=\msingle{(\e,\sigma)}
\] 
\end{proposition}

For machine evaluation, which does not return a substitution explicitly, we may obtain $\sigma$ by populating the initial stack with the free variables $\val{x_1} \dots \val{x_n}$ in $E$. This is the domain of $\sigma$; the machine then substitutes into each variable $\val{x_i}$ to return $\sigma \val{x_i}$, as follows, so that the relation from input to output stack captures the MGU $\sigma$.
\[
	\run 
	  {\val{x_1}~\dots~\val{x_n}} {{\color{black}} \llbracket `E{\color{black}}\rrbracket} \e
	  {\sigma \val{x_1}~\dots~\sigma \val{x_n}} {*} \e
\]


\subsection{Kappa-calculus}\label{kappa-encoding}
Hasegawa's $\kappa$-calculus~\cite{Hasegawa:decomposing-LC},  which featured in the introduction, as formulated by Power and Thielecke~\cite{Power:kappa-cat}, embeds as follows.
\[
\begin{array}{r@{~~\Coloneqq~~}c@{~~\mid~~}c@{~~\mid~~}c}
	E,E' & ~E;E' & \mathsf{PUSH}\,x & \kappa x.E
\\[5pt]
	\term{M,N} & ~\term{M;N} & \term{[x]} & \term{Ex.<x>;M}
\end{array}
\]
A $\kappa$-abstraction (first-order $\lambda$-abstraction) decomposes into a \emph{new} $\term{Ex}$ and a \emph{pop} $\term{<x>}$, so that $\kappa x$ acts as a binder as well as taking input. The asymmetry of this construct, which rules out terms such as $\term{<x>;<x>;[x]}$, gives the Cartesian semantics. The machine behaviour illustrates that the decomposition gives the correct behaviour: pop and substitute.
\[
\begin{array}{@{(~}l@{~,~}r@{~,~}r@{~)}}%
         S\,\val z & \term{Ex.(<x>;M)}   & K %
\\\hline S\,\val z & \term{<y>;\{y/x\}M} & K %
\\\hline S\,\val z & \term{<y>}          & \term{(\{y/x\}M)}\,K %
\\\hline S         & \term{*}            & \term{(\{z/x\}M)}\,K %
\\\hline S         & \term{\{z/x\}M}     & K %
\end{array}%
\]
Note that Power and Thielecke's formulation of the $\kappa$-calculus,  above,  unfortunately does not include the unit for sequencing; an omission which is fixed in the FMC and RMC. 

\subsection{Pattern-matching}

In a $\lambda$-calculus with \emph{pattern-matching}~\cite{Kesner:pm,Cirstea-Kirchner-1998,klop_lambda_2008}, abstractions are over \emph{patterns}, which generally are algebraic terms, instead of over variables. Since we do not have higher-order abstraction, we will stay within $\kappa$-calculus. We will use \emph{pairing} as our only pattern, as is standard in a Cartesian setting.
\[
	E,E'~\Coloneqq~
	      E;E'
	~\mid~\mathsf{PUSH}~p
	~\mid~\kappa p.E
\qquad
	p,q~\Coloneqq~
	      x
	~\mid~(p,q)
\]
In $\kappa p.E$, the free variables of the pattern $p$ bind in $E$. The embedding in the RMC captures this by introducing each variable as \emph{new}:
\[
\begin{array}{@{}r@{}}
	\term{M,N}~\Coloneqq~
	      \term{M;N}
	~\mid~\term{[p]}
	~\mid~\term{E!x.<p>;M} \qquad(\val{!x} = \fv{p})
\end{array}
\]
The behaviour of the RMC-encoding is slightly more general than that of the original pattern-matching calculus, since it can deal with situations such as $\term{[x];Ey.Ez.<(y,z)>}$ (by substituting into $x$), which the latter cannot. Pattern-matching calculi address this by introducing explicit product types. With that constraint, the overall behaviour agrees, as expected.

Languages using \emph{let-bindings} with patterns, used for string diagrams or monoidal categories~\cite{Hasegawa-1997,Power:kappa-cat}, admit an analogous encoding to that of monadic let-bindings in the FMC~\cite{Heijltjes:FMCI}, as follows.
\[
	\mathsf{let}~E=p~\mathsf{in}~E' \quad\mapsto\quad \term{`E ; E!x.<p> ; `E'}\quad (\val{!x} = \fv{p})
\]


\subsection{Symmetric pattern-matching}

We embed a first-order fragment of \emph{Theseus},  a \emph{reversible} programming language which is both forwards and backwards deterministic~\cite{Sabry:symmetric-pattern-matching,choudhury_symmetries_2022,Chardonnet:curry-howard-reversible}. Its \emph{values} are generated from pairing and injections. Computation is performed by application of \emph{isomorphisms},  expressed using \emph{symmetric pattern-matching} syntax below,  where the set of values $v_i$ (respectively,  $w_i$) are restricted by the type system to be \emph{exhaustive} and \emph{non-overlapping},  guaranteeing reversibility.  We omit presentation of the type system here for brevity. 
\[
	\{ ~v_1 \leftrightarrow w_1 ~\mid~ \ldots ~\mid~ v_n \leftrightarrow w_n~\} 
\]
We interpret values using corresponding function symbols,  and embed isomorphisms as below.
\[
	\term{\exists !x_1.<v_1>;[w_1] ~+~ \ldots ~+~ \exists !x_n.<v_n>;[w_n]}
\]
Each $\val{!x_i}$ contains exactly the variables of ${v_i}$ and $ {w_i}$.  Isomorphism \emph{inversion} is given by $(-)^\dagger$.  The application of isomorphism $\term{M}$ to value $\val v$ is encoded as $\term{[v];M}$.


\subsection{Prolog}

\newcommand\prologcolon{\mathrel{{:}{:}{=}}}
\newcommand\query{\mathrel{{?}{-}}}
\newcommand\inet[2]{\langle#1\mid#2\rangle}
\newcommand{\ssep}{\, ;\,}

We embed pure Prolog,  as defined by the following grammar.
\[
	\begin{array}{@{}l@{\quad}r@{~}l@{\quad}l@{\quad}r@{~}l@{}}
	\text{Terms:}     & t \coloneqq&  x \mid f(t_1, \ldots, t_n) & \text{Atoms:}     & A \coloneqq&  P(t_1, \ldots, t_n) \\
	\text{Clauses:}   & C \coloneqq&  A ~{:}{-}~ A_1\dots A_n  & \text{Programs:}  & L \coloneqq&  C_1\dots  C_n
	\end{array}
\]
Prolog by defualt evaluates using a \emph{top-down} evaluation strategy,  which takes the query --- that is,  a chosen atom --- as a \emph{goal} to be proved and non-deterministically applies clauses,  decomposing the set of goals into sub-goals,  succeeding when the set is empty.  

We consider atoms,  as well as terms,  to be modelled in the RMC by algebraic terms. 
The embedding of a clause $C$ with free variables $\rvecup x$ is then given as follows.
\[
{\llbracket A ~{:}{-}~ A_1\dots A_n\rrbracket} ~=~ \term{E!x . <A>;[A_1]; .. ;[A_n]}
\]
The stack will hold the set of goals,  with the non-deterministic sum of clauses modelling the program itself. To select the correct solution among the non-deterministic outputs,  we add a new constant $\val{`{\mathsf{END}}}$ to the input stack,  which can only be removed when the computation is successfully completed.  
The embedding of a query-program pair is thus given below.
\[
	\llbracket (Q, C_1 \ldots C_n) \rrbracket
	=
	\term{[`{\mathsf{END}}];[Q];\,(}\llbracket C_1\rrbracket \term{+ \ldots +~} \llbracket C_n\rrbracket \term{)^*;<`{\mathsf{END}}>}
\]
As we did with unification,  by populating the initial stack with the query $Q$ we collect the substitutions generated by the machine run.  The completed computation then returns the same instantiation of $Q$ as does Prolog.

\begin{proposition}[Prolog embeds]
\label{prop:prolog}
Given as input a query $Q$ and program $L$,  the Prolog abstract interpreter~\cite{art-of-prolog} outputs the instance $Q'$  if and only if 
\[
	\run{\val{Q}}{\term{{\color{black}\llbracket (Q,L)\rrbracket}}}{\e}{\val{Q'}}{\term{\star}}{\e}
\]
\end{proposition}

In contrast to top-down evaluation,  \emph{bottom-up} evaluation starts from the set of \emph{facts} (\emph{i.e.},  clauses with a head $A$ but no body $A_i$) and non-deterministically applies rules \emph{in reverse},  building up the set of logical consequences of facts,  succeeding when the query is reached.  
Remarkably,  the RMC reveals a syntactic duality between the two strategies: the dual embedding ${\llbracket (Q,L)\rrbracket^{\dagger}}$ evaluates the query-program pair via the bottom-up strategy.  

Further,  although beyond the scope of the current work,  we believe that \emph{semi-naive evaluation} --- which can asymptotically improve the performance of bottom-up queries --- can be formalized in our language as a source-to-source translation. 


\section{Equational Theory and Reduction}
\label{sec:equational-theory}

In this section,  we define the equational theory of the RMC and show that it is sound with respect to a natural notion of observational equivalence induced by the operational semantics.  We then recover a reduction relation through a natural orientation of (an appropriate subset of) the equations,  which we prove confluent.  In the star-free fragment,  this relation is strongly normalising,  giving rise to normal forms.


\subsection{The Equational Theory}

The equational theory we define is conservative over Kleene algebra,  but contains the following additional axioms.  For the \emph{new} variable construct,  there are axioms reminiscent of those from \emph{nominal} Kleene algebra \cite{Gabbay:nominal-KA}.  There are then axioms internalizing unification,  with each corresponding to a step in the standard Martelli and Montanari algorithm for unification (including its "occurs check")~\cite{martelli_efficient_1982}. In particular,  these axioms include a symmetric version of the $\beta$-law of the $\kappa$-calculus.  There is also a weaker,  inequational version of the $\eta$-law; and we shall see later that,  in the typed setting,  we recover the stronger,  standard $\eta$-law.  Finally,  there is an axiom allowing the permutation of locations,  familiar from the Functional Machine Calculus \cite{Heijltjes:FMCI}. Observe that the equational theory respects the duality of the calculus,  as expected.

\begin{definition}[Equational theory]
We define the \emph{equational theory} $(=)$ over terms of the RMC to be the least congruence generated by the following axioms. 

\begin{itemize}
	\item 
Those of Kleene algebra: that is,  $(\term{M},  \term{;},  \term{*},  \term{+},  \term{0})$ forms an idempotent semiring\footnote{That is,  $(\term{M}, \term{;}, \term{*})$ is a monoid and $(\term{M},  \term{+},  \term{0})$ is an idempotent commutative monoid,  such that $\term{;}$ distributes over $\term{+}$ and $\term{0}$ is annihilative. }
together with:
\begin{align*}
	\term{M^*} = \term{* \,+\, M;M^*}
	&&
	 \term{M;N} \leq \term{N} \to \term{M^*;N} \leq \term{N} \\
	\term{M^*} = \term{*\,+\, M^*;M} 
	&&
	\term{M;N} \leq \term{M} \to \term{M;N^*} \leq \term{M} 
\end{align*}
where $\leq$ is the natural partial order: $\term{M} \leq \term{N}$ iff $\term{M+N} = \term{N}$. 
In the sequel,  we work modulo associativity wherever possible. 
\item Those axiomatizing the \emph{new} variable constructor: 
\[
\begin{array}{l@{~}l@{\quad}l}
    \term{Ex.M}        &=_\nu  \term{M}         & (\val x \notin \fv{\term M})
\\  \term{M ; (Ex.N)}  &=_\nu  \term{Ex.M;N}  & (\val x \notin \fv{\term M})
\\  \term{(Ex.M) ; N}  &=_\nu  \term{Ex.M;N}  & (\val x \notin \fv{\term N})
\\	\term{Ex.(M + N)}  &=_\nu  \term{Ex.M ~+~ Ex.N}
\\  \term{Ex.Ey.M}     &=_\nu  \term{Ey.Ex.M}
\end{array}
\]
\item Those axiomatizing $\beta$- and $\eta$-equivalence and unification:
\begin{align*}
    \term{Ex. N;[t]a;a<x>;M}  &=_\beta    \term{\{t/x\}(N;M)} & (\term x \not\in \val{t})
\\  \term{Ex. N;[x]a;a<t>;M}  &=_\beta    \term{\{t/x\}(N;M)} & (\term x \not\in \val{t})
\\  \term{Ex. a<x>;[x]a}      &\leq_\eta  \term{*}
\\  \term{a[x];<x>a}          &=_\upsilon \term{*}
\\  \term{[f(S)]a;a<f(T)>}    &=_\upsilon \term{[S]a;a<T>}
\\  \term{[f(S)]a;a<g(T)>}    &=_\upsilon \term{0}            & ( f \neq  g)
\\  \term{[f(S)]a;a<x>}       &=_\omega   \term{0}            & (\val  x \in S)
\\  \term{[x]a;a<f(S)>}       &=_\omega   \term{0}            & (\val  x \in S) 
\end{align*}
\item Those axiomatizing permutation of locations:
\begin{align*}
	\term{[s]a;b<t>} =_\pi \term{b<t>;[s]a} && ( a \neq  b)
\end{align*}
\end{itemize}
\end{definition}

The permutation law suffices to derive the fact that "push" and "pop" on distinct locations commute with each other. 

\begin{lemma}[Permutation]
\label{lem:permutation}
The following equations are derivable.
\begin{align*}
	\term{[s]a;[t]b} = \term{[t]b;[s]a} \quad
	 \term{a<s>;b<t>} = \term{b<t>;a<s>} \quad ( a \neq  b)
\end{align*}
\end{lemma}

\begin{proof}
Derivable from $\pi$ together with $\eta, \nu$ and $\beta$.
\end{proof}

A $\beta$-redex of the $\kappa$-calculus encodes as below, left,  assuming a chosen,  unnamed location.  As such,  its $\beta$-law decomposes into the following two axioms (where $\val x$ is not free in $\term t$).
\[
	\term{[t]~;~Ex.<x>;M} ~=_\nu~ \term{Ex.[t];<x>;M} ~=_\beta~ \term{\{t/x\}M}
\]
The $\beta$-axiom of the RMC generalizes and symmetrizes the second equality above: the $\term{Ex}$ construct may additionally be separated from the sub-term $\term{[t];<x>}$ by some term $\term{N}$,  which is also open for substitution.

We now define a notion of observational equivalence of terms,  which we call \emph{machine equivalence}, via the big-step operational semantics.  
In the following,  we take $S_A$ to be open in general (i.e.\ it may contain free variables),  and possibly sharing variables with $\term M$.  Open memories and substitutions are considered equivalent up to a choice of globally fresh variables.  Practically,  this means machine runs and memories are equipped with a context; we omit this for brevity.

\begin{definition}[Machine equivalence]
 \emph{Machine equivalence} is the relation $(\mequiv)$ on terms defined by
\[
	\term{M} \mequiv \term{N} \quad \text{if} \quad  \forall S_A.~\evalset{S_A}{M} =  \evalset{S_A}{N}~.
\]
\end{definition}

\begin{proposition}
\label{lem:congruence}
Machine equivalence is a congruence. 
\end{proposition}

We can define a similar \emph{machine refinement} $(\lesssim)$ relation on terms as follows. 
\[
	\term{M} \lesssim \term{N} \quad \text{if} \quad \forall S_A.~\evalset{S_A}{M} \subseteq \evalset{S_A}{N}~
\]
It is easy to see that machine refinement satisfies the property that $\term{M} \lesssim \term{N}$ if and only if $\term{M + N} \mequiv \term{N}$, and so $(\lesssim)$ stands in the same relation to $(\mequiv)$ as $(\leq)$ does to $(=)$ in the algebraic theory.  In particular,  it is a partial order closed under all contexts.

\begin{theorem}[Soundness of equational theory]
\label{thm:equations-sound}
The equational theory is sound with respect to machine equivalence: 
\[
	\term{M} = \term{N} ~ \Rightarrow ~ \term{M} \sim \term{N} ~.
\]
\end{theorem}


\subsection{Reduction and Confluence}

It is possible to recover a reduction relation $(\rw)$ from the equational theory by orientation of the equations. 
The $\beta$ and $\upsilon$ equations have an obvious orientation and the $\nu$ equations are oriented to bring $\term{Ex}$ to the front of a term,  or remove it when it does not bind any variables. We do not consider the $\eta$ inequation in the untyped case. 
To maintain confluence in the presence of the Kleene star,  we find it necessary to break the symmetry of the equational theory: we can include either the left-handed \emph{unfolding} reduction $\term{M^*} \rw \term{* ~+~ M;M^*}$
or its dual $\term{M^*} \rw \term{* ~+~ M^*;M}$,  but not both.  We opt,  arbitrarily,  for the former.  We do not find it sensible to consider the conditional equation $\term{M;N} \leq \term{N} \to \term{M^*;N} \leq \term{N}$ or its dual as part of reduction. 

We consider the main interest of confluence for reduction to be in dealing with the directed $\beta$, $\upsilon$ and $\nu$ equations.  We thus make several further adaptations to the KA fragment of reduction which simplify the confluence proof.  We choose to omit the {idempotence} equation, in keeping with the more general multiset operational semantics; we omit  consideration of associativity and commutativity of $(\term{+})$,  although the method of proof should extend straightforwardly to these cases; and we assume right-associativity of sequencing,  so terms are of the form $\term{M;(N;\ldots;(P;*)\ldots)}$, which also has the pleasant consequence of eliminating the need to consider $\term{(Ex.M ); N} \rw \term{Ex.(M;N)}$ where $\val x \not\in \fv{\term N}$ and $\term{M;*} \rw \term{M}$. 
For this reason, $\beta$-reduction will take place in (also right-associative) \emph{linear contexts}, defined as follows
\[
    \term{L\{~\}} \Coloneqq \term{\{~\}} ~\mid~ \term{M;L\{~\}} ~\mid~ \term{L\{~\};M}
\]
and $\upsilon$-reductions will contain an arbitrary post-fixed term.  We also enforce a prioritization of left- over right-distributivity, although we  expect this can be lifted.  Finally, the permutation equation is oriented left--to--right.  The reduction relation is overall clearly \emph{sound} (but not {complete}) with respect to the equational theory,  and thus the operational semantics.

\begin{definition}[Reduction]
\label{def:reduction} 
The reduction relation $(\rw)$ on terms is given by the rewrite rules in Figure~\ref{fig:reduction}, closed under all contexts.
\end{definition}

\begin{figure}
\[
\label{eq:a}
\begin{array}{@{\quad}r@{~~}l@{~\,}l@{~}l}
\multicolumn{4}{@{}l@{}}{\bullet~~\text{Kleene algebra-reduction $(\rw_\KA)$:}}
\\[5pt]
    \term{M^*}      & \rw_\KA & \term{*~+~ M;M^*} 
\\  \term{(M+N);P}  & \rw_\KA & \term{M;P ~+~ N;P} & (\term P\not = \term{M+N})    
\\  \term{P;(M+N)}  & \rw_\KA & \term{P;M ~+~ P;N}
\\  \term{M+0}      & \rw_\KA & \term{M}
\\  \term{0+M}      & \rw_\KA & \term{M}
\\  \term{*;M}      & \rw_\KA & \term{M}
\\  \term{M;0}      & \rw_\KA & \term{0}           
\\  \term{0;M}      & \rw_\KA & \term{0}
\\[5pt]
\multicolumn{4}{@{}l@{}}{\bullet~~\text{New variable-reduction $(\rw_\nu)$:}}
\\[5pt]
    \term{Ex.M}     & \rw_\nu & \term{M}                     & (\val x\notin \term M)
\\  \term{M;Ex.N}   & \rw_\nu & \term{Ex.M;N}                & (\val x\notin \term M)
\\  \term{Ex.(M+N)} & \rw_\nu & \term{Ex.M ~+~ Ex.N}
\\[5pt]
\multicolumn{4}{@{}l@{}}{\bullet~~\beta\text{- and unification-reduction $(\rw_\beta,\rw_\upsilon)$:}}
\\[5pt]
	\term{Ex. L\{[t]a;a<x>;M\}}  & \rw_\beta & \term{Ex. \{t/x\} L\{M\}} & (\val x\notin \term t)
\\  \term{Ex. L\{[x]a;a<t>;M\}}  & \rw_\beta & \term{Ex. \{t/x\} L\{M\}} & (\val x\notin \term t)
\\  \term{[x]a;a<x>;M}           & \rw_\upsilon & \term{M}
\\  \term{[f(T)]a;a<f(S)>;M}     & \rw_\upsilon & \term{[T]a;a<S>;M}
\\  \term{[f(T)]a;a<g(S)>;M}     & \rw_\upsilon & \term{0}               & ( f\not =  g)
\\  \term{[f(S)]a;a<x>;M}        & \rw_\upsilon & \term{0}               & (\val x \in S)
\\  \term{[x]a;a<f(S)>;M}        & \rw_\upsilon & \term{0}               & (\val x \in S)
\\[5pt]
\multicolumn{4}{@{}l@{}}{\bullet~~\text{Permutation-reduction $(\rw_\pi)$:}}
\\[5pt]
	\term{[t]a;b<s>;M}            & \rw_\pi & \term{b<s>;[t]a;M}           & (a\neq b)
\end{array}
\]
\caption{Reduction rules of the RMC}
\label{fig:reduction}
\Description[Reduction rules of the RMC]{}
\end{figure}

Our proof of confluence makes use of techniques from the field of first-order term rewriting by treating the atomic terms $\term{[t]a}$, $\term{a<t>}$, $\term{\star}$ and $\term{0}$ as constants and binders $\term{Ex}$ as distinct function symbols for each $\val x$.  Local confluence is shown by analysis of critical pairs. 

\begin{lemma}
\label{localconfluence} 
The reduction relation $(\rw)$ is locally confluent.
\end{lemma}

We now observe the star-free fragment is strongly normalising (by the measure given in the proof of Proposition \ref{lem:big-step-well-defined}) and hence confluent by Newman's Lemma \cite{newman_lemma_1942}.  Confluence of the full reduction relation then follows from an application of the Hindley-Rosen Lemma \cite{hindley_lemma_1964}.

\begin{theorem}[Confluence]
\label{theorem:confluence}
Reduction $(\rw)$ is confluent. 
\end{theorem}

In the (strongly normalising) star-free fragment,  we thus recover the existence of \emph{normal forms}.  
Let $(\rw^*)$ be the reflexive, transitive closure of $(\rw)$. 
We state the result for the case of one location,  but in general \emph{push} and \emph{pop} actions on distinct locations can be arranged according to a chosen ordering.  

\begin{corollary}[Normal forms]
For any closed,  star-free RMC-term $\term{M}$,  we have that $\term{M} \rw^* \term{N}$,  where
\[
	\term{N} \equiv \term{N_1 + \ldots + N_n}~, 
\]
where each $\term{N_i}$ is sum-free and of the form
\[
	\term{N_i} \equiv \term{E!x.<s_n>;\ldots;<s_1>;[t_1];\ldots;[t_m]}~, 
\]
where $\val{!x}$ are the variables of $\val{s_j}$ and $\val{t_k}$.
\end{corollary}

Note,  normal forms are taken modulo associativity and commutativity of $(\term{+})$ and permutation of $\term{Ex}$. We close this section with a theorem about the equational theory (and thus about operational semantics) proved easily using the existence of normal forms. 

\begin{theorem}[Equations for duality] 
\label{prop:converse}
We have the following:
\begin{itemize}
\item for any closed,  star-free term $\term{M}$,  we have $\term{M} \leq \term{M;M}^\dagger\term{;M} $;
\item if $\term M$ is also sum-free then $\term{M} = \term{M;M}^\dagger\term{;M}$.
\end{itemize}
\end{theorem}

The first statement is familiar as the unique law added to KA in Kleene algebra \emph{with converse}~\cite{Pous:KA-with-converse};  the second statement is familiar as the weakening of invertibility in the definition of \emph{inverse monoids} and \emph{inverse cateogries}~\cite{Cockett:restriction-categories}.


\section{The Simply-Typed RMC}
\label{sec:simple-types}

In this section,  we introduce a simple typing discipline giving the input and output arity of terms, making their stack use explicit. This results in a natural language for string diagrams where types represent the input and output wires, which we explicate in the subsequent section on categorical semantics.  

In a first-order setting, the power of types is naturally limited. Types do not confer termination: the fragment without Kleene star is inherently terminating, while the full RMC is non-terminating also when typed.  Nor do types prevent failure, as in the term $\term 0$ or failure of unification, and this is semantically correct: failure represents the empty relation, and so should not be excluded from the typed calculus. 

What types do enforce is a notion of \emph{progress}: they guarantee that there are sufficient elements on the stack for the machine to continue. In particular, they prevent terms with Kleene star from consuming (or producing) an arbitrary number of stack elements. These constraints are essential for correct relational composition, and hence in giving a relational semantics to terms, where types represent the source and target of a relation. Specifically, types allow the sound strengthening the $\eta$-axiom to an equivalence, which is necessary for a monoidal categorical semantics.

Formally, the type of a term $\term M$ will give the input and output arity \emph{for each location}. With a single, unnamed location, we will have $\term{M:m>n}$ for natural numbers $m$ and $n$; for a set of locations $A$, we will have an input arity $m_a$ and output arity $n_a$ for each $a\in A$.

\begin{definition}[Simple types]
A \emph{memory type} $\type{n_A}=\{n_a\mid a\in A\}$ is a family of natural numbers in a (finite) set of locations $A$. A \emph{type} $\type{m_A > n_A}$ consists of an \emph{input} and an \emph{output} memory type.
\end{definition}

We may consider a memory type $\type{n_A}$ as a function from locations to natural numbers, where $\type{n_A}(a)=n_a$ for $a\in A$ and $\type{n_A}(b)=0$ for any $b\notin A$, and silently expand $\type{n_A}$ to $\type{n_B}$ for any $B\supset A$. For any natural number $n$ we may write $\type{n_a}$ for the singleton family in $\{a\}$. The \emph{empty} memory type $\type{0_A}$ (or simply $\type0$) is everywhere zero, and the \emph{sum} $\type{n_A+m_A}$ of two memory types is given location-wise by $\type{(n_A+m_A)}(a)=\type{n_A}(a)+\type{m_A}(a)$.

\begin{definition}[The simply-typed RMC]
The typing rules for simply-typed RMC-terms $\term{M:m_A>n_A}$ are given in Figure~\ref{fig:simple-types}.
\end{definition}

\newcommand\tskip{\mathsf{skip}}
\newcommand\tseq{\mathsf{seq}}
\newcommand\tzero{\mathsf{zero}}
\newcommand\tsum{\mathsf{sum}}
\newcommand\tpush{\mathsf{push}}
\newcommand\tpop{\mathsf{pop}}
\newcommand\tnew{\mathsf{new}}
\newcommand\tstar{\mathsf{star}}

\newcommand\texp{\mathsf{exp}}
\newcommand\tdual{\mathsf{dual}}
\newcommand\srule[1]{\!\!\scriptstyle{#1}}

\begin{figure}
\[
\begin{array}{cc}
	\infer[\srule\tskip]{\term{* : n_A > n_A}}{}
&	\infer[\srule\tseq] {\term{M;N : k_A > n_A}}{\term{M: k_A > m_A} && \term{N: m_A > n_A}}
\\ \\
	\infer[\srule\tzero]{\term{0 : m_A > n_A}}{}
&	\infer[\srule\tsum] {\term{M+N : m_A > n_A}}{\term{M: m_A > n_A} && \term{N: m_A > n_A}}
\\ \\
	\infer[\srule\tstar]{\term{M^* : n_A > n_A}}{\term{M: n_A > n_A}}
&	\infer[\srule\tnew] {\term{Ex.M : m_A > n_A}}{\term{M: m_A > n_A}}	
\\ \\
	\infer[\srule\tpush]{\term{[t]a : n_A > n_A + 1_a}}{}
&	\infer[\srule\tpop] {\term{a<t> : 1_a + n_A > n_A}}{}
\end{array}
\]
\caption{The simply-typed RMC}
\label{fig:simple-types}
\Description[The simply-typed RMC]{}
\end{figure}
 
The present type system, which records only the \emph{number} of elements consumed and produced on the stack, is much simplified with respect to that of the FMC~\cite{Heijltjes:FMCI}. Being higher-order, the latter records also the \emph{types} of stack elements, as a vector of types $\type{!t}$. Here, we only need the length of $\type{!t}$ as a natural number $n$; because the RMC is first-order, and also because we have a single-sorted signature $\Sigma$ for algebraic terms, by which we assume only a single base type for values. A many-sorted signature, where each sort is represented by a different base type, may be accommodated with stack types $\type{!t}$ as on the FMC.

We have the following three basic properties. First, \emph{subject reduction} (reduction preserves types). Second, \emph{expansion}: if the machine runs with a given input stack, it will also run with a larger stack, and the output will be equally larger. Third, \emph{duality} exchanges input and output types.

\begin{lemma}
\label{lem:typing-judgements}
The following properties of the type system hold: 
\begin{itemize}
	\item \emph{Subject reduction:}
if $\term{M:k_A > n_A}$ and $\term M\rw\term N$ then also \\$\term{N:k_A > n_A}$.

	\item \emph{Expansion:}
if $\term{M:k_A > n_A}$ then $\term{M:k_A+m_A > m_A + n_A}$.

	\item \emph{Duality:}
if $\term{M:k_A > n_A}$ then $\term{M^\dagger: n_A > k_A}$.
\end{itemize}
\end{lemma}

To connect types with machine behaviour, we extend types to stacks and memories in the expected way: a stack is typed $S:\type{n}$ where $n=|S|$, the length of $S$, and a memory is typed $S_A:\type{n_A}$ where $S_a:\type{n_a}$ for all $a\in A$. Continuation stacks and states are then typed by the following rules.
\[
\begin{array}{c}
	   \infer{\e:\type{n_A > n_A}}{}
\qquad \infer{\term M\,K:\type{k_A > n_A}}{\term{M:k_A>m_A} & K:\type{m_A>n_A}}
\\\\   \infer{(S_A,\term M,K):\type{0>n_A}}{S_A:\type{k_A} & \term{M:k_A>m_A} & K:\type{m_A>n_A}}
\end{array}
\]

\begin{theorem}[The machine respects types]
\label{thm:machine-respects-types}
For a run
\[
	\run {S_A}MK {T_A}NL
\]
if $\mstate{S_A}MK:\type{0>n_A}$ then $\mstate{T_A}NL:\type{0>n_A}$.
\end{theorem}

A machine state $\mstate{S_A}{M}{K}$ makes \emph{progress} unless $\term M$ is a pop $\term{a<t>}$ and $S_a$ is empty. Note that types rule out the latter configuration, and thus ensure progress. Since types are preserved by the above theorem, we have the following corollary.

\begin{corollary}[Machine progress]
\label{cor:machine-progress}
Any state on a run from a typed state makes progress.
\end{corollary}

We define a typed notion of machine equivalence, where we expect input stacks to respect typing.

\begin{definition}[Typed machine equivalence]
\emph{Typed machine equivalence} is defined as the relation $\term{M}\mequiv\term{N:m_A>n_A}$
 on similarly typed terms, which holds if: 
\[
	 \forall S_A:\type{m_A}.~ \evalset{S_A}{M} =  \evalset{S_A}{N}~.
\]
\end{definition}

We can strengthen the $\eta$-axiom for the \emph{typed} equational theory, as it ensures that the stack involved is non-empty.

\begin{definition}[Typed equational theory]
The \emph{typed equational theory} $\term M=\term{N:m_A>n_A}$ relates similarly typed terms of the RMC by the least congruence generated by the axioms of the untyped equational theory, together with the strengthened $\eta$-axiom 
\[
	 \term{Ex.a<x>;[x]a} =_\eta \term{* : 1_a+m_A > m_A + 1_a}~.
\]
\end{definition}

The $\eta$-law of the $\kappa$-calculus is decomposed as follows, where $\val x \not\in \fv{\term M}$, and indeed our symmetric $\eta$-axiom is equivalent.
\[
	       \term{Ex.<x>;[x];M} 
~ =_\nu  ~ \term{(Ex.<x>;[x]); M} 
~ =_\eta ~ \term{*;M} 
~ =      ~ \term{M}
\]

\begin{theorem}[Soundness of typed equational theory] \label{thm:typed-equations-sound}
The typed equational theory is sound for typed machine equivalence: 
\[
	\term{M} = \term{N: m_A > n_A} \quad\Longrightarrow\quad \term{M} \sim \term{N: m_A > n_A}~.
\]
\end{theorem}

\section{Categorical semantics}\label{sec:categorical-semantics}
We give a sound and complete categorical semantics for the sum-free,  star-free and function-symbol-free fragment of the RMC,  illustrating how the core calculus provides a term language and operational interpretation of Frobenius monoids,  and thus of a natural class of string diagrams.  
Following this,  we show that the corresponding fragments with linear and "Cartesian" variable policies provide a term language for symmetric monoidal and Cartesian categories,  respectively. 
The extension of the categorical semantics to include non-deterministic sum follows easily,  but we leave to future work the categorical interpretation of iteration and of algebraic terms.  


Frobenius monoids 
are used as a categorical abstraction for modelling a range of computational phenomena,  many of interest to programming language theorists.  
While the literature contains computational interpretations of compact closed categories \cite{chen_computational_2021} and certain accounts of operational semantics of string diagrams \cite{bonchi_bialgebraic_2021, zanasi:ctxt-signal-flow},  to the best knowledge of the authors,  the RMC gives the first syntactic account of Frobenius monoids via a notion of $\beta$-reduction,  and the first operational account of Frobenius monoids via a stack machine.   
To highlight the correspondence with string diagrams,  we work with a single location,  but the results extend easily to multiple locations \cite{Barrett:thesis}. 
%

\paragraph{Notation} 
Given a category $\mathcal{C}$ and objects $X,Y \in \mathcal{C}$ we denote by $\mathcal{C}(X,Y)$ the corresponding homset.  We write the identity morphism on $X$ as $\id_X$,  or simply $X$.  Composition is written in diagrammatic order,  using infix $(;\!)$.  We denote the tensor product of a \emph{symmetric monoidal category (SMC)} $\mathcal{C}$ by $\otimes$,  its unit by $I$ and its symmetry natural transformation as $\sym$ \cite{MacLane}.  Components of natural transformations are indexed by subscripts, which we sometimes omit. We elide all associativity and unit isomorphisms associated with monoidal categories and sometimes.  Recall that a commutative monoid in a given SMC is an object $X$ together with a multiplication $\mu: X \otimes X \to X$ and a unit $\eta: I \to X$ satisfying the expected associativity,  commutativity and unit equations,  and a cocommutative comonoid is defined dually.  We denote by $\textsf{CMon}$ the category of commutative monoids and monoid homomorphisms.   We will use string diagrams throughout this section to describe (co-)monoids.  

\begin{definition}[Frobenius monoid]
In an SMC $(\mathcal{C}, \otimes, 1)$, an \emph{extra-special commutative Frobenius monoid} $(X, \delta,  \epsilon,  \mu,  \eta)$ consists of a commutative monoid $(X, \mu, \eta)$ and a cocommutative comonoid ($X, \delta, \epsilon)$ that satisfy the \emph{Frobenius},  \emph{special},  and \emph{extra} equations,  given below in the string-diagrammatic language of Figure \ref{fig:frobenius-generators}.
\[
	\begin{tikzpicture}
	\begin{pgfonlayer}{nodelayer}
		\node [style=none] (110) at (5.5, 0.75) {$=$};
		\node [style=none] (119) at (6, 0.75) {};
		\node [style=none] (120) at (7.5, 0.75) {};
		\node [style=none] (121) at (-7.5, -1) {};
		\node [style=vertex] (133) at (-6.75, 0.5) {};
		\node [style=none] (134) at (-7.5, 0.5) {};
		\node [style=none] (135) at (-5.25, 1) {};
		\node [style=none] (136) at (-6, -0.5) {};
		\node [style=vertex] (137) at (-6, -0.5) {};
		\node [style=none] (138) at (-5.25, -0.5) {};
		\node [style=vertex] (141) at (-3.5, 0) {};
		\node [style=none] (142) at (-2.5, 0) {};
		\node [style=none] (143) at (-4, 0.5) {};
		\node [style=none] (144) at (-4, -0.5) {};
		\node [style=vertex] (145) at (-2.5, 0) {};
		\node [style=none] (146) at (-3.5, 0) {};
		\node [style=none] (147) at (-2, 0.5) {};
		\node [style=none] (148) at (-2, -0.5) {};
		\node [style=none] (157) at (-4.75, 0) {$=$};
		\node [style=vertex] (158) at (4.25, 0.75) {};
		\node [style=none] (159) at (5, 0.75) {};
		\node [style=none] (160) at (3.5, 1.25) {};
		\node [style=none] (161) at (3.5, 0.25) {};
		\node [style=vertex] (162) at (2.75, 0.75) {};
		\node [style=none] (163) at (2, 0.75) {};
		\node [style=none] (164) at (3.5, 1.25) {};
		\node [style=none] (165) at (3.5, 0.25) {};
		\node [style=vertex] (166) at (2.75, -1) {};
		\node [style=vertex] (167) at (4.25, -1) {};
		\node [style=none] (168) at (5.5, -1) {$=$};
		\node [style=none] (169) at (6.75, -1) {$\textsf{id}_I$};
		\node [style=none] (170) at (1.25, -1) {};
		\node [style=vertex] (171) at (0.5, 0.5) {};
		\node [style=none] (172) at (1.25, 0.5) {};
		\node [style=none] (173) at (-1, 1) {};
		\node [style=none] (174) at (-0.25, -0.5) {};
		\node [style=vertex] (175) at (-0.25, -0.5) {};
		\node [style=none] (176) at (-1, -0.5) {};
		\node [style=none] (177) at (-1.5, 0) {$=$};
	\end{pgfonlayer}
	\begin{pgfonlayer}{edgelayer}
		\draw (119.center) to (120.center);
		\draw [in=45, out=180] (135.center) to (133);
		\draw [in=120, out=-60] (133) to (136.center);
		\draw (133) to (134.center);
		\draw (137) to (138.center);
		\draw [in=135, out=0] (143.center) to (141);
		\draw [in=0, out=-135] (141) to (144.center);
		\draw (141) to (142.center);
		\draw [in=45, out=180] (147.center) to (145);
		\draw [in=180, out=-45] (145) to (148.center);
		\draw (145) to (146.center);
		\draw [in=135, out=0] (160.center) to (158);
		\draw [in=0, out=-135] (158) to (161.center);
		\draw (158) to (159.center);
		\draw [in=45, out=180] (164.center) to (162);
		\draw [in=180, out=-45] (162) to (165.center);
		\draw (162) to (163.center);
		\draw [in=0, out=-135] (137) to (121.center);
		\draw (167) to (166);
		\draw [in=135, out=0] (173.center) to (171);
		\draw [in=60, out=-120] (171) to (174.center);
		\draw (171) to (172.center);
		\draw (175) to (176.center);
		\draw [in=180, out=-45] (175) to (170.center);
	\end{pgfonlayer}
\end{tikzpicture}
\]
\end{definition}

\begin{figure*}
\begin{tikzpicture}
	\begin{pgfonlayer}{nodelayer}
		\node [style=vertex] (0) at (24.75, -1.5) {};
		\node [style=none] (1) at (25.75, -1.5) {};
		\node [style=none] (2) at (23.75, -1) {};
		\node [style=none] (3) at (23.75, -2) {};
		\node [style=none] (4) at (24.75, -3) {$\mu: 2 \to 1$};
		\node [style=none] (5) at (33.5, -1.5) {};
		\node [style=vertex] (6) at (32.5, -1.5) {};
		\node [style=none] (7) at (33.25, -3) {$\eta: 0 \to 1$};
		\node [style=vertex] (9) at (7.5, -1.5) {};
		\node [style=none] (10) at (6.5, -1.5) {};
		\node [style=none] (11) at (8.5, -1) {};
		\node [style=none] (12) at (8.5, -2) {};
		\node [style=none] (13) at (15.75, -1.5) {};
		\node [style=vertex] (14) at (16.75, -1.5) {};
		\node [style=none] (15) at (16.25, -3) {$\epsilon: 1 \to 0$};
		\node [style=none] (16) at (7.5, -3) {$\delta: 1 \to 2$};
		\node [style=none] (17) at (7.5, -4) {$\term{\exists x.<x>;[x].[x]: 1 > 2}$};
		\node [style=none] (18) at (16.25, -4) {$\term{\exists x.<x>: n > 0}$};
		\node [style=none] (19) at (24.75, -4) {$\term{\exists x.<x>;<x>;[x]: 2 > 1}$};
		\node [style=none] (20) at (33.25, -4) {$\term{\exists x.[x]: 0 > 1}$};
		\node [style=termbox] (21) at (6.5, 4.25) {$\term{M}$};
		\node [style=termbox] (22) at (8.5, 4.25) {$\term{N}$};
		\node [style=none] (23) at (7, 4.5) {};
		\node [style=none] (24) at (8, 4.5) {};
		\node [style=none] (31) at (7.5, 1) {$\term{M;N: k > n}$};
		\node [style=none] (32) at (7.5, 2) {$M;N: k \to n$};
		\node [style=none] (33) at (31.75, 5) {};
		\node [style=none] (34) at (31.75, 3.75) {};
		\node [style=none] (35) at (34.75, 5) {};
		\node [style=none] (36) at (34.75, 3.75) {};
		\node [style=none] (51) at (16.25, 1) {$\term{M: k+m > k+n}$};
		\node [style=none] (52) at (24.75, 1) {$\term{\exists !x.<?x>;M;[!x]: m+k > n+k}$};
		\node [style=none] (53) at (33.25, 1) {$\term{\exists x {\kern1pt} y.<x>;<y>;[x];[y]: 2 > 2}$};
		\node [style=none] (54) at (16.25, 2) {$k \otimes M: k \otimes m \to k \otimes n$};
		\node [style=none] (55) at (24.75, 2) {$M \otimes k: m \otimes k \to n \otimes k$};
		\node [style=none] (56) at (33.25, 2) {$\sym: 2 \to 2$};
		\node [style=none] (59) at (7, 4) {};
		\node [style=none] (60) at (8, 4) {};
		\node [style=none] (101) at (9, 4.5) {};
		\node [style=none] (102) at (10.25, 4.5) {};
		\node [style=none] (103) at (9, 4) {};
		\node [style=none] (104) at (10, 4) {};
		\node [style=none] (105) at (10.75, 4.25) {$\type{n}$};
		\node [style=none] (106) at (6, 4.5) {};
		\node [style=none] (107) at (4.75, 4.5) {};
		\node [style=none] (108) at (6, 4) {};
		\node [style=none] (109) at (5, 4) {};
		\node [style=none] (110) at (4.25, 4.25) {$\type{k}$};
		\node [style=termbox] (112) at (16.25, 5) {$\term{M}$};
		\node [style=none] (113) at (16.75, 5.25) {};
		\node [style=none] (114) at (16.75, 4.75) {};
		\node [style=none] (115) at (15.75, 5.25) {};
		\node [style=none] (116) at (14.25, 5.25) {};
		\node [style=none] (117) at (15.75, 4.75) {};
		\node [style=none] (118) at (14.5, 4.75) {};
		\node [style=none] (119) at (13.75, 5) {$\type{m}$};
		\node [style=none] (120) at (18.25, 5.25) {};
		\node [style=none] (121) at (16.75, 5.25) {};
		\node [style=none] (122) at (18, 4.75) {};
		\node [style=none] (123) at (16.75, 4.75) {};
		\node [style=none] (124) at (14.25, 4) {};
		\node [style=none] (125) at (14.25, 3.5) {};
		\node [style=none] (126) at (17.75, 4) {};
		\node [style=none] (127) at (14.75, 4) {};
		\node [style=none] (128) at (17.5, 3.5) {};
		\node [style=none] (129) at (15, 3.5) {};
		\node [style=none] (130) at (18.75, 5) {$\type{n}$};
		\node [style=termbox] (131) at (24.75, 3.75) {$\term{M}$};
		\node [style=none] (132) at (25.25, 4) {};
		\node [style=none] (133) at (25.25, 3.5) {};
		\node [style=none] (134) at (24.25, 4) {};
		\node [style=none] (135) at (23.25, 4) {};
		\node [style=none] (136) at (24.25, 3.5) {};
		\node [style=none] (137) at (23.5, 3.5) {};
		\node [style=none] (138) at (22.75, 3.75) {$\type{m}$};
		\node [style=none] (139) at (26.25, 4) {};
		\node [style=none] (140) at (25.25, 4) {};
		\node [style=none] (141) at (26, 3.5) {};
		\node [style=none] (142) at (25.25, 3.5) {};
		\node [style=none] (143) at (26.75, 3.75) {$\type{n}$};
		\node [style=none] (144) at (22.75, 5.25) {};
		\node [style=none] (145) at (22.75, 4.75) {};
		\node [style=none] (146) at (26.75, 5.25) {};
		\node [style=none] (147) at (22.75, 5.25) {};
		\node [style=none] (148) at (26.5, 4.75) {};
		\node [style=none] (149) at (23, 4.75) {};
		\node [style=none] (150) at (14.25, 3.75) {$\type{k}$};
		\node [style=none] (151) at (18.25, 3.75) {$\type{k}$};
		\node [style=none] (152) at (27.25, 5) {$\type{k}$};
		\node [style=none] (153) at (22.25, 5) {$\type{k}$};
	\end{pgfonlayer}
	\begin{pgfonlayer}{edgelayer}
		\draw [in=135, out=0] (2.center) to (0);
		\draw [in=0, out=-135] (0) to (3.center);
		\draw (0) to (1.center);
		\draw (6) to (5.center);
		\draw [in=45, out=180] (11.center) to (9);
		\draw [in=180, out=-45] (9) to (12.center);
		\draw (9) to (10.center);
		\draw (14) to (13.center);
		\draw [in=0, out=180] (36.center) to (33.center);
		\draw [in=-180, out=0, looseness=1.25] (34.center) to (35.center);
		\draw (23.center) to (24.center);
		\draw (59.center) to (60.center);
		\draw (101.center) to (102.center);
		\draw (103.center) to (104.center);
		\draw (107.center) to (106.center);
		\draw (109.center) to (108.center);
		\draw (116.center) to (115.center);
		\draw (118.center) to (117.center);
		\draw (121.center) to (120.center);
		\draw (123.center) to (122.center);
		\draw (127.center) to (126.center);
		\draw (129.center) to (128.center);
		\draw (135.center) to (134.center);
		\draw (137.center) to (136.center);
		\draw (140.center) to (139.center);
		\draw (142.center) to (141.center);
		\draw (147.center) to (146.center);
		\draw (149.center) to (148.center);
	\end{pgfonlayer}
\end{tikzpicture}
\caption{Equipment of the Relational Machine Category: string diagrams,  categorical combinators,  and RMC-terms. }
\label{fig:frobenius-generators}
\Description[Equipment of the Relational Machine Category: string diagrams,  categorical combinators,  and RMC-terms. ]{}
\end{figure*}
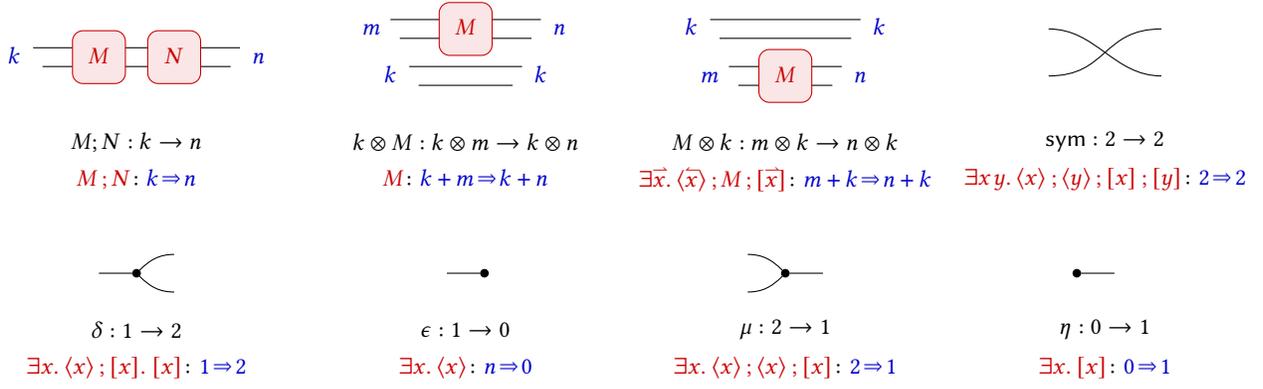
%

The categorical abstraction of relations which is modelled by the RMC is that of \emph{extra-hypergraph categories},  as defined below.  This is a slight variant of the usual definition of \emph{hypergraph category} \cite{fong_hypergraph_2019},  which additionally requires that the relevant Frobenius structures satisfy the \emph{extra} equation.  
Indeed, a paradigmatic example of such a category is that of sets and relations. 
\begin{definition}\label{def:hypergraph-category}[Extra-hypergraph category]
An \emph{extra-hypergraph category} is an SMC $\mathcal{C}$ in which each object $X \in \mathcal{C}$ is equipped with an extra-special commutative Frobenius structure $(\delta_X,  \epsilon_X,  \mu_X,  \eta_X)$, which together satisfy the following coherences:
\begin{align*}
	\delta_{X \otimes Y} &= (\delta_X \otimes \delta_Y) ; (X \otimes \sym_{X,Y} \otimes Y) &
	\epsilon_{X \otimes Y} &= \epsilon_X \otimes \epsilon_Y \\
	\mu_{X \otimes Y} &= (X \otimes \sym_{X,Y} \otimes Y) ; (\mu_A \otimes \mu_Y) &
	\eta_{X \otimes Y} &= \eta_X \otimes \eta_Y 
\end{align*}
and the unit coherence that $\eta_I = \id_I = \epsilon_I$.  
A functor of extra-hypergraph categories is a strong symmetric monoidal functor that additionally preserves the Frobenius structures. 
\end{definition}

We can define the hypergraph categorical structure of the RMC by assigning terms to string diagrams as in Figure \ref{fig:frobenius-generators}.  We take objects to be the natural numbers,  and wires to represent the input and output stacks, with the head of the stack at the top. 
In contrast to the $\lambda$-calculus,  categorical composition of terms is given by sequential compostion $(\term{;})$ rather than by substitution.  
Right-action $ k \otimes {M}$ of the tensor product is given by stack expansion (as described in Lemma \ref{lem:typing-judgements}). 
Left-action ${M} \otimes k$ lifts the $k$ arguments from the stack as variables $\val{!x}$,  to restore them after evaluating ${M}$ on the remaining stack.  The remaining equipment is familiar from the discussion in the introduction of this paper. 
\begin{definition}\label{def:rmc-category}
The \emph{Relational Machine Category}, $\termcategory{\Sigma}$,  of RMC-terms,  is defined by the following data. 
\begin{itemize}
\item Objects: natural numbers; 
\item Morphisms: $\termcategory{\Sigma}(\type m, \type n)$ is the set of closed simply-typed terms $ \term{M: m > n}$ over $\Sigma$,  modulo $(=)$;  
\item Identity: at object $\type{n}$,  the term $\term{*: n > n}$; 
\item Composition: for terms $\term{M: k > m}$ and $\term{N: m > n}$,  their composition is $\term{M;N: k > n}$; 
\item Extra-hypergraph structure: the tensor $\otimes$ acts on objects $\type{m} \otimes \type{n} = \type{m + n}$, with unit $I = \type{0}$; 
 the right- and left-action of $\otimes$ on morphisms,  the symmetry,  and the Frobenius monoids on object $\type 1$ are given by Figure \ref{fig:frobenius-generators};  {associators} and unitors are given by the identity.
\end{itemize}
\end{definition}
Note that it suffices to define a Frobenius monoid on the generating object $\type 1$,  with Frobenius monoids at $\type n$ defined using the compatibility conditions of Definition \ref{def:hypergraph-category}.

\begin{theorem}[Soundness]\label{prop:hypergraph-category}
The Relational Machine Category \\$\termcategory{\Sigma}$ is an extra-hypergraph category.
\end{theorem}
It is immediate from the axioms of Kleene algebra that $\termcategory{\Sigma}$ is also \emph{\textsf{CMon}-enriched},  \emph{i.e.}, every homset $\termcategory{\Sigma}(\type{n}, \type{m})$ is a commutative monoid with addition $(\term{+})$ and unit $\term{0}$,  such that  composition of morphisms distributes over addition.  This is an easy and appropriate extension of the model to incorporate non-determinism.  

Completeness of the categorical semantics for the core fragment follows an easy normal form argument. 
Let $\emptyset$ denote the empty signature.    
Here,  and below,  let the free categories in question be generated over a single object (and the empty set of morphisms) unless otherwise specified.  
\begin{theorem}[Hypergraph completeness]\label{thm:completeness}
The subcategory of sum-free,  star-free terms of $\termcategory{\emptyset}$ is equivalent to the free extra-hypergraph category.  
\end{theorem}

By restricting the variable policy of this fragment appropriately,  we similarly recover term languages for symmetric monoidal and Cartesian categories.  Let \textit{linear} RMC-terms to be those for which each occurring variable appears exactly once in a \textit{push} and once in a \textit{pop} to its left.  Then,  since linear RMC-terms containing no function symbols are exactly permutations,  we have the following. 
\begin{theorem}[Linear completeness]
The subcategory of linear,  sum-free,  star-free terms of $\termcategory{\emptyset}$ is equivalent to the free symetric monoidal category.  
\end{theorem}

We similarly take the \textit{Cartesian} RMC-terms to be those for which each bound variable occurs exactly once in a \textit{pop},   possibly in many \textit{pushes} to its right,  {and} in which no \textit{pop} contains anything other than a variable. 
Observe further that an RMC-signature $\Sigma$ is just an algebraic signature with no equations,  which can be used to generate a free Cartesian category as its Lawvere theory.  Then a simple normal form argument proves the following. 
\begin{theorem}[Cartesian completeness]
The subcategory of sum-free,  star-free terms of $\termcategory{\Sigma}$ is equivalent to the free Cartesian category generated over $\Sigma$. 
\end{theorem}
Consideration of the case where the signature includes {an equational theory} is left for future work: although it is unproblematic in the Cartesian case,  the general case would require working with unification \textit{modulo a theory}.

We conclude this section by giving some simple examples of RMC-terms and their corresponding string diagrams.  
Following existing work on string diagrams for (the positive fragment of) Tarski's calculus of relations \cite{tarski:1941,Bonchi:graphical-conjunctive-queries,bonchi:tape-diagrams, zanasi:logic-programming,  bonchi_graphical_2019},  we give terms and diagrams whose anticipated relational denotation is given by \emph{top} (\emph{i.e.},  everything is related), \emph{intersection}, and \emph{converse}.  The remaining operations of the calculus of relations are already included in the grammar of regular expressions (see subsection \ref{subsec:regex}). 
Top ($\top$) and intersection ($\cap$) are encoded as below.  
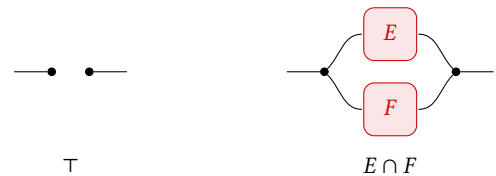
\[
	\begin{tikzpicture}
	\begin{pgfonlayer}{nodelayer}
		\node [style=vertex] (0) at (-4.25, 0) {};
		\node [style=none] (1) at (-3.25, 0) {};
		\node [style=none] (2) at (-5.25, 1) {};
		\node [style=none] (3) at (-5.25, -1) {};
		\node [style=none] (5) at (-13, 0) {};
		\node [style=vertex] (6) at (-14, 0) {};
		\node [style=vertex] (9) at (-7.75, 0) {};
		\node [style=none] (10) at (-8.75, 0) {};
		\node [style=none] (11) at (-6.75, 1) {};
		\node [style=none] (12) at (-6.75, -1) {};
		\node [style=none] (13) at (-16, 0) {};
		\node [style=vertex] (14) at (-15, 0) {};
		\node [style=none] (18) at (-15.5, -3.5) {$\term{\exists x.\exists y.<x>;[y]}$};
		\node [style=none] (19) at (-7, -3.5) {$\term{\exists x.\exists y.<x>;[x];`E;<y>;[x];F;<y>;[y]}$};
		\node [style=termbox] (22) at (-6, -1) {$\term{F}$};
		\node [style=termbox] (53) at (-6, 1) {$\term{`E}$};
		\node [style=none] (103) at (-14.5, -2.5) {$\top$};
		\node [style=none] (104) at (-6, -2.5) {$E \cap F$};
	\end{pgfonlayer}
	\begin{pgfonlayer}{edgelayer}
		\draw [in=135, out=0] (2.center) to (0);
		\draw [in=0, out=-135] (0) to (3.center);
		\draw (0) to (1.center);
		\draw (6) to (5.center);
		\draw [in=45, out=180] (11.center) to (9);
		\draw [in=180, out=-45] (9) to (12.center);
		\draw (9) to (10.center);
		\draw (14) to (13.center);
	\end{pgfonlayer}
\end{tikzpicture}
\]
We can further characterize the converse of a typed RMC-term as follows,  first noting that 
the compositions $\eta ; \delta$ and $\mu ; \epsilon$ form the cups and caps of a (self-dual) compact closed structure.  By $\beta$-reduction we have that $\eta ; \delta = \term{E!x.[!x];[!x]}$,  and dually,  as shown below. 
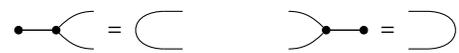
\[
\begin{array}{c@{\qquad\qquad}c}
    \begin{tikzpicture}
    	\begin{pgfonlayer}{nodelayer}
    		\node [style=none] (5) at (5, 0) {};
    		\node [style=vertex] (6) at (4, 0) {};
    		\node [style=vertex] (9) at (5, 0) {};
    		\node [style=none] (10) at (4, 0) {};
    		\node [style=none] (11) at (6, 0.5) {};
    		\node [style=none] (12) at (6, -0.5) {};
    	\end{pgfonlayer}
    	\begin{pgfonlayer}{edgelayer}
    		\draw (6) to (5.center);
    		\draw [in=45, out=180] (11.center) to (9);
    		\draw [in=180, out=-45] (9) to (12.center);
    		\draw (9) to (10.center);
    	\end{pgfonlayer}
    \end{tikzpicture}
~=~
    \begin{tikzpicture}
    	\begin{pgfonlayer}{nodelayer}
    		\node [style=none] (18) at (12.25, 0.5) {};
    		\node [style=none] (19) at (12.25, -0.5) {};
    		\node [style=none] (20) at (11.5, 0) {};
    		\node [style=none] (23) at (12.75, 0.5) {};
    		\node [style=none] (24) at (12.75, -0.5) {};
    	\end{pgfonlayer}
    	\begin{pgfonlayer}{edgelayer}
    		\draw [in=90, out=180] (18.center) to (20.center);
    		\draw [in=-180, out=-90] (20.center) to (19.center);
    		\draw (23.center) to (18.center);
    		\draw (19.center) to (24.center);
    	\end{pgfonlayer}
    \end{tikzpicture}
&
    \begin{tikzpicture}
    	\begin{pgfonlayer}{nodelayer}
    		\node [style=none] (5) at (-5, 0) {};
    		\node [style=vertex] (6) at (-4, 0) {};
    		\node [style=vertex] (9) at (-5, 0) {};
    		\node [style=none] (10) at (-4, 0) {};
    		\node [style=none] (11) at (-6, 0.5) {};
    		\node [style=none] (12) at (-6, -0.5) {};
    	\end{pgfonlayer}
    	\begin{pgfonlayer}{edgelayer}
    		\draw (6) to (5.center);
    		\draw [in=135, out=0] (11.center) to (9);
    		\draw [in=0, out=-135] (9) to (12.center);
    		\draw (9) to (10.center);
    	\end{pgfonlayer}
    \end{tikzpicture}
~=~
    \begin{tikzpicture}
    	\begin{pgfonlayer}{nodelayer}
    		\node [style=none] (18) at (-12.25, 0.5) {};
    		\node [style=none] (19) at (-12.25, -0.5) {};
    		\node [style=none] (20) at (-11.5, 0) {};
    		\node [style=none] (23) at (-12.75, 0.5) {};
    		\node [style=none] (24) at (-12.75, -0.5) {};
    	\end{pgfonlayer}
    	\begin{pgfonlayer}{edgelayer}
    		\draw [in=90, out=0] (18.center) to (20.center);
    		\draw [in=0, out=-90] (20.center) to (19.center);
    		\draw (23.center) to (18.center);
    		\draw (19.center) to (24.center);
    	\end{pgfonlayer}
    \end{tikzpicture}
\\ \\
	\term{Ex.[x];[x]: 0 > 2}
& 	\term{Ex.<x>;<x>: 2 > 0}
\end{array}
\]
Relational converse can now be internalized in the following way. 
\begin{proposition}[Typed Duality]\label{prop:algebraic-duality}
For closed $\term{M: m > n}$ we have $\term{M}^\dagger ~=~ \term{\exists !x{\kern1pt}!y.<?x>;[!y];M;<?x>;[!y] : n > m}$,  where $|\val{!x}| = n$, $|\val{!y}| = m$.
\end{proposition}
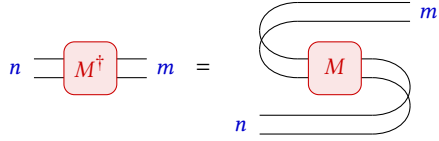
\[
	\begin{tikzpicture}
	\begin{pgfonlayer}{nodelayer}
		\node [style=termbox] (0) at (-0.5, 0) {$\term{M}$};
		\node [style=none] (1) at (0, 0.25) {};
		\node [style=none] (2) at (0, -0.25) {};
		\node [style=none] (3) at (-1, 0.25) {};
		\node [style=none] (4) at (-1.5, 0.25) {};
		\node [style=none] (5) at (-1, -0.25) {};
		\node [style=none] (6) at (-1.5, -0.25) {};
		\node [style=none] (7) at (2, 1.5) {$\type{m}$};
		\node [style=none] (8) at (0.5, 0.25) {};
		\node [style=none] (9) at (0, 0.25) {};
		\node [style=none] (10) at (0.5, -0.25) {};
		\node [style=none] (11) at (0, -0.25) {};
		\node [style=none] (12) at (-3, -1.5) {$\type{n}$};
		\node [style=none] (13) at (1.5, 1.25) {};
		\node [style=none] (14) at (1.5, 1.75) {};
		\node [style=none] (15) at (-2.5, -1.25) {};
		\node [style=none] (16) at (-2.5, -1.75) {};
		\node [style=none] (17) at (0.5, -1.75) {};
		\node [style=none] (18) at (0.5, -1.25) {};
		\node [style=none] (19) at (-1.5, 1.25) {};
		\node [style=none] (20) at (-1.5, 1.75) {};
		\node [style=none] (21) at (-2.5, 1) {};
		\node [style=none] (22) at (-2.5, 0.5) {};
		\node [style=none] (23) at (1.5, -0.5) {};
		\node [style=none] (24) at (1.5, -1) {};
		\node [style=termbox] (25) at (-7, 0) {$\term{M^\dagger}$};
		\node [style=none] (26) at (-6.5, 0.25) {};
		\node [style=none] (27) at (-6.5, -0.25) {};
		\node [style=none] (28) at (-7.5, 0.25) {};
		\node [style=none] (29) at (-8.5, 0.25) {};
		\node [style=none] (30) at (-7.5, -0.25) {};
		\node [style=none] (31) at (-8.5, -0.25) {};
		\node [style=none] (32) at (-5.5, 0.25) {};
		\node [style=none] (33) at (-6.5, 0.25) {};
		\node [style=none] (34) at (-5.5, -0.25) {};
		\node [style=none] (35) at (-6.5, -0.25) {};
		\node [style=none] (36) at (-5, 0) {$\type{m}$};
		\node [style=none] (37) at (-9, 0) {$\type{n}$};
		\node [style=none] (38) at (-4, 0) {$=$};
		\node [style=none] (39) at (-1.5, -0.25) {};
		\node [style=none] (40) at (-1.5, 1.25) {};
		\node [style=none] (41) at (-2.5, 0.5) {};
		\node [style=none] (42) at (-2.5, 0) {};
		\node [style=none] (43) at (0.5, -1.25) {};
		\node [style=none] (44) at (0.5, 0.25) {};
		\node [style=none] (45) at (1.5, -0.5) {};
		\node [style=none] (46) at (1.5, -1) {};
		\node [style=none] (47) at (0.5, -1.75) {};
		\node [style=none] (48) at (0.5, -0.25) {};
		\node [style=none] (49) at (1.5, -1) {};
		\node [style=none] (50) at (1.5, -1.5) {};
	\end{pgfonlayer}
	\begin{pgfonlayer}{edgelayer}
		\draw (4.center) to (3.center);
		\draw (6.center) to (5.center);
		\draw (9.center) to (8.center);
		\draw (11.center) to (10.center);
		\draw (14.center) to (20.center);
		\draw (19.center) to (13.center);
		\draw (18.center) to (15.center);
		\draw (16.center) to (17.center);
		\draw [in=-90, out=180] (4.center) to (21.center);
		\draw [in=180, out=90] (21.center) to (20.center);
		\draw (29.center) to (28.center);
		\draw (31.center) to (30.center);
		\draw (33.center) to (32.center);
		\draw (35.center) to (34.center);
		\draw [in=-90, out=180] (39.center) to (41.center);
		\draw [in=180, out=90] (41.center) to (40.center);
		\draw [in=-90, out=0] (43.center) to (45.center);
		\draw [in=0, out=90] (45.center) to (44.center);
		\draw [in=-90, out=0] (47.center) to (49.center);
		\draw [in=0, out=90] (49.center) to (48.center);
	\end{pgfonlayer}
\end{tikzpicture}
\]  
\section{Relational Semantics}\label{sec:relational-semantics}
In this section,  we investigate in detail the relational semantics of the typed RMC.  
Although the denotational semantics is very close to the input/output behaviour of the big-step semantics,  there are some differences: forgetting the output substitution,  and the  interpretation of $\term{Ex.M}$ and $\term{<t>}$. The former is interpreted semantically as the union of all possible interpretations of $\val{x}$ as some closed term $\val{t}$,  leaving the latter to become a simple equality check predicate; there is no need for unification since free variables are dealt with by (possibly infinite) non-determinism of the interpretation of $\term{Ex}$. 

For simplicity,  we again work with a single location,  but the results generalize easily.  
Now,  let us recall the extra-hypergraph structure of the category of sets and relations. 
\begin{definition}
The category \relcat~ of sets and relations forms an extra-hypergraph category where the tensor product is given by the Cartesian product,  the unit given by the singleton set $\{\e\}$,  and a Frobenius monoid on each object $A$ given by $\delta_A =   \{ (x, (x,x)) ~|~ x \in A\}$, $\epsilon_A = \{ (x , \e) ~|~ x \in A\}$ and with $\mu_A$ and $\eta_A$ by their respective converses. 
\end{definition}

The relational semantics of typed RMC-terms can easily be read off their normal forms: for example,  compare the relational semantics of $\delta, \epsilon, \mu$ and $\eta$ with the RMC-terms interpreting the same morphisms in Figure \ref{fig:frobenius-generators}.  Confirming this,  we define an extra-hypergraph functor $\rel{-}: \termcategory{\Sigma} \to \relcat$.  
Let $T0$ be the set of \emph{closed} algebraic terms built over $\Sigma$,  \emph{i.e.} containing no variables.  
To interpret open terms,  we define a valuation $v: \textsf{Var} \to T0$ as a partial function on the set of variables $\textsf{Var}$; the valuation $v\{\val{x} \leftarrow \val t\}$ assigns $\val t$ to $\val x$ and otherwise behaves as $v$.  We elide the isomorphism between closed algebraic terms $\val t$ in the syntax of RMC-terms and $\val t \in T0$. 
We extend the action of $v$ to algebraic terms with variables $\Gamma$ so that $v(\val{f(t_1, \ldots, t_n)}) = \val{f(}v(\val {t_1})\val{,\ldots,}v(\val{t_n})\val{)}$. 
In the following,  we consider relations $A \to B$ as functions $A \to \mathcal{P}(B)$,  where $\mathcal{P}$ is the powerset functor,  and we will write $f^n: X \to X$ for the n-fold composition of a function $f: X \to X$ and $X^n$ for the n-fold product of a set $X$. 
\begin{definition}[Relational semantics] 
The \emph{relational semantics} of the typed RMC is defined inductively on types by $\rel{\type n} = T0^n$.
We define the relation $\rel{\term{M: m > n}}_v$ inductively on the type derivation of terms as follows,  and may abbreviate this as $\relterm{M}_v$.  
\begin{align*}
	\rel{\term{*: n > n}}_v(S) &= [ S ]  \\
	\rel{ \term{M;N: k > n}}_v(S) &= [ U ~|~ T \in \relterm{M}_v(S) ,~ U\in \relterm{N}_v(T) ] \\
	\rel{ \term{0: m > n}}_v(S) &= [\,] \\
	\rel{ \term{M+N: m > n}}_v(S) &= \relterm{M}_v(S) \sqcup \relterm{N}_v(S) \\
	\rel{ \term{[t]: n > n+1}}_v(S) &= [ S \, v(\val t) ] \\
	\rel{ \term{<t>: n+1 > n}}_v(S\,  \val u) &= \begin{cases} 
		[ S ] & \text{if}~ v(\val t) = \val u \\
		[ ] & \text{otherwise}
		\end{cases} \\
	\rel{ \term{Ex.M: m > n}}_v(S) &= \bigsqcup_{\val t \in T0}\rel{ \term{M}}_{v\{\val{x} \leftarrow \val t\}}(S) \\
	\rel{ \term{M^*: n > n}}_v(S) &= \bigsqcup_{n \in \mathbb N} \rel{  \term{M}}_v^n(S)
\end{align*}
Given a closed term $\term{M}$,  its relational semantics $\rel{\term{M}}$ (written without a subscript) is given by $\lfloor\, \rel{\term{M}}_\e \rfloor$ ,  where $\e$ is the empty valuation.
\end{definition}


From the following lemma,  soundness of the semantics follows. 
\begin{lemma}[Substitution]\label{lem:substitution}
For any $\term{M: m > n}$,  we have 
\begin{align*}
	\rel{ \term{M: m > n}}_{v\{ \val x \leftarrow \val t\}}  =  
	\rel{ \term{\{t/x\}M: m > n}}_{v} ~.
\end{align*}
\end{lemma}
\begin{proposition}[Hypergraph functor] \label{denotation-well-defined}
The relational semantics $\rel{-}: \termcategory{\Sigma} \to \relcat$ is a well-defined extra-hypergraph functor. 
\end{proposition}
In fact,  if we additionally consider $\termcategory{\Sigma}$ and $\relcat$ as \textsf{CMon}-enriched categories,  with the monoid on homsets in $\relcat$ given by the union of relations,  it is easy to see the relational semantics functor is also \textsf{CMon}-enriched functor.

\section{Encoding computational models, II}
\label{sec:encoding-models-ii}

We conclude with further encodings of computational models, demonstrating the use of locations, an innovation of the Functional Machine Calculus \cite{Heijltjes:FMCI, Heijltjes:FMCII} that may capture stateful behaviour.  

\subsection{Guarded Command Language}

In 1975 Dijkstra introduced the \textit{Guarded Command Language (GCL)}: a simple non-deterministic imperative programming language \cite{dijkstra-GCL-1975}. 

The eponymous \textit{guarded command} $B \to S$ is given by a sequence of statements $S$ to be evaluated when the \textit{guard} Boolean expression $B$  is true.  Statements are atomic actions --- here,  we consider writing to a memory cell ($a := N$) --- or are sets of guarded commands wrapped in \textit{conditional} (\textsf{if}) or \textit{iteration} (\textsf{do}) keywords.  The former construct executes (non-deterministically) one of the statements whose guard is true,  and aborts execution if none are.  The iteration construct successively executes (non-deterministically) one of the statements whose guard is true,  and terminates successfully when none are.  Reading from a memory cell $a$ is denoted $!a$.  
The full grammar of the GCL is given below: 
\[
\begin{array}{@{}l@{\quad}l@{~}l@{}}	
	\text{Boolean Expressions:} & B & \Coloneqq \textsf{tt} \mid \textsf{ff} \mid \ldots\\
	\text{Numeric Expressions:} & N & \Coloneqq n\mid \,!a \mid \ldots \\
	\text{Guarded Commands:}    & C & \Coloneqq \textsf{abort} \mid C_1 \square\,  C_2 \mid B \rightarrow S  \\
	\text{Statements:}          & S & \Coloneqq \textsf{skip} \mid S_1;S_2 \mid \textsf{if}~C \mid \textsf{do} ~C \mid a \coloneqq N
\end{array}
\]
where $n \in \mathbb Z$ and $a$ ranges over a set of variable names. \!\footnote{Note,  we make several simplifications with regards to the original grammar.  In particular,  guarded commands and lists of statements are allowed to be empty.  We consider only a meagre set of expressions,  but, with a little more work,  \textit{e.g.} operators such as addition or comparison may be accounted for.  }
Input and output is performed by manipulation of memory cells,  which we take to hold integers --- so,  in standard imperative style,  programs are called for their side-effects. 
%

The embedding of the GCL in the RMC is given below. 
\[
\begin{aligned}
	\embed{ \textsf{tt}    } &= \term{[\top]}      & \embed{!a}                  &= \term{\exists x. a<x>;[x]a;[x]}
\\  \embed{ \textsf{ff}    } &= \term{[\bot]}      & \embed{ a \coloneqq N     } &= \embed{N} \term{; Exy. a<x>;<y>;[y]a}
\\  \embed{          n     } &= \term{[n]}         & \embed{ B \rightarrow S   } &= \embed{B} \term{;<\top>;} \embed{S}
\\ 	\embed{ \textsf{abort} } &= \term{0}           & \embed{ C_1\,\square\,C_2 } &= \embed{C_1}\,\term{+}\,\embed{C_2}
\\  \embed{ \textsf{skip}  } &= \term{*}           & \embed{ S_1 ; S_2         } &= \embed{S_1}  \term{;}  \embed{S_2}
\\	\embed{ \textsf{if}~C  } &= \embed{C}          & \embed{ \textsf{do}~C     } &= \embed{C}\term{^*}	
\end{aligned}
\]
Encodings of global memory cells are familiar from the FMC \cite{Heijltjes:FMCI}: locations model mutable variables, whose value is held as the single item on the corresponding stack.  The result of an expression,  including reading from a memory cell $(!a)$,  is pushed to the main (unnamed) stack; updating a memory cell $(a := N)$ replaces the value stored at $a$ with the value pushed to the main stack by $\llbracket N \rrbracket$. 

\begin{proposition}[The RMC implements GCL]
If the GCL successfully executes a guarded command $C$ on input memory $S_A$, with output memory $T_A$,  then there exists an RMC run
\[
\run{S_A}{`{{\color{black}\llbracket \mathcal{C}\rrbracket}}}{\e}{T_A}{*}{\e}~. 
\]
\end{proposition}
Note that our encoding in fact outputs (non-deterministically) the memory at every stage of the computation,  and not only the memories at the end of a successful GCL execution.

\subsection{Turing Machines}

We demonstrate an encoding of a Turing machine $\mathcal M$ with a set of \emph{states} $S$,  initial state ${i}\in S$,  halting state ${h}\in S$,  alphabet $\Sigma$,  blank symbol ${0}\in\Sigma$,  and transition function $\delta : (S\setminus \{{h}\}) \times \Sigma \to S \times \Sigma\times \{\boldsymbol{L}, \boldsymbol{R}\}$.  
To do so, we work with an RMC signature consisting of $S \cup \Sigma$ and three locations $ l, r,$ and $ q$,  used to record the tape to the left of the head,  the tape under and to the right of the head (in reverse),  and the current state,  respectively. 

A Turing machine $\mathcal M$ then encodes as follows. Below left are the encodings of tape moves $m\in\{\boldsymbol{L}, \boldsymbol{R}\}$ and of a transition in $\delta$ as a five-tuple $d=(s,a,s',a',m)$. Then $\mathcal M$ encodes as an iterated non-deterministic sum of transitions, below right.
\[
\begin{aligned}
   \llbracket\boldsymbol{L}\rrbracket &= \term{Ex.l<x>;[x]r}
\\ \llbracket\boldsymbol{R}\rrbracket &= \term{Ex.r<x>;[x]l}
\\ \llbracket d\rrbracket &= \term{q<s>;r<a>;[s']q;[a']r;}\llbracket m\rrbracket
\end{aligned}
\qquad
	\llbracket \mathcal M \rrbracket = {\left(\sum_{d\in\delta}\llbracket d\rrbracket\right)^\term\ast}
\]
Instantiating the tape in both directions $S_l$ and $S_r$ with streams of blank symbols\footnote{It also easy enough to simulate an infinite tape with stacks,  using special end-of-stack markers $\val \bot$ and $\val \top$.  The encoding of a read of $a\in \Sigma$ then additionally checks for the end of either stack,  pushing markers along as needed: that is,  in the case of $a = 0$,  using $\term{r<a> + r<\bot>;[\bot]l + r<\top>;[\top]r}$ instead of $\term{r<a>}$.}, we have the following result.

\begin{proposition}[The RMC implements Turing machines]
A Turing machine $\mathcal M$ halts with tape $T_l$ and $T_r$ to the left and right of the head, respectively,  if and only if there exists a machine run
\[
\run{S_l  \cdot S_r \cdot  \val{s}_q }{`{{\color{black}\llbracket \mathcal{M}\rrbracket}}}{\e}{T_l \cdot T_r \cdot  \, \val h_q}{*}{\e}~. 
\]
\end{proposition}

\subsection{Interaction Nets}

Interaction nets are a graphical model of computation first introduced as a generalisation of multiplicative proof nets  \cite{lafont_interaction_1990}.  A term calculus for interaction nets~\cite{Fernandez-Mackie-1999}  is given as follows.   A \emph{net} is a pair $\inet{T}{\Delta}$ where $T$ is a stack of algebraic terms over a signature $\Sigma$, and $\Delta=\{t_1\doteq u_1,\dots,t_n \doteq u_n\}$ is a set of formal equations.  Each variable occurs at most twice in a net. Rewriting of nets is performed according to a set of rules $\mathcal R$,  with each rule a pair of terms $f(S)\bowtie g(U)$ such that $f\neq g$ and with each variable occurring exactly twice. The set $\mathcal R$ is further required to be \emph{deterministic}: there is at most one rule for each pair of $f$ and $g$; and \emph{symmetric}: if $f(S)\bowtie g(U)$ then $g(U)\bowtie f(S)$. 
Variables are \emph{local} to each rule, and are instantiated with fresh variables during computation.

The evaluation relation $(\rw_{\mathcal R})$ is generated by the following rules, adapted from~\cite{Fernandez-Mackie-1999}.  We write $T \doteq U$ for the pairwise equations of $T$ and $U$ of equal length, and $\e$ the empty set of equations. 
\[
\begin{aligned}
	\inet{T}{\Delta,x\doteq u}       & ~\rw_{\mathcal R}~ \inet{\{u/x\}T}{\{u/x\}\Delta}
\\	\inet{T}{\Delta,u\doteq x}       & ~\rw_{\mathcal R}~ \inet{\{u/x\}T}{\{u/x\}\Delta}
\\	\inet{T}{\Delta,f(R)\doteq g(V)} & ~\rw_{\mathcal R}~ \inet{T}{\Delta, R\doteq S,U\doteq V }
\quad (f(S)\bowtie g(U))
\end{aligned}
\]
To encode interaction nets we use two locations $ l, r$ to hold the left and right sides of each equation in $\Delta = \{t_1 \doteq u_1, \dots ,t_n \doteq u_n\}$, which is then interpreted as the memory below left. Below it is the interpretation of a rule in $\mathcal R$, where $\val{!x}$ are the variables occurring in $S$ and $U$. A set of rules $\mathcal R$ is encoded as below right. The terms $T$ in a configuration $\inet T\Delta$ are held on the unnamed main location.
\[
\begin{aligned}
	\llbracket \{t_i \doteq u_1\}_{i\leq n} \rrbracket &= (\val{t_1}\dots\val{t_n})_l\cdot(\val{u_1}\dots\val{u_n})_r 
\\
	\llbracket f(S)\bowtie g(U) \rrbracket &= \term{E!x.l<f(S)>;r<g(U)>}
\end{aligned}
\quad
\llbracket\mathcal R\rrbracket = \left(\sum_{r\in\mathcal R}\llbracket r\rrbracket\right)^\term\ast
\]
 
\begin{proposition}[Interaction nets embed]
For a set of rules $\mathcal R$ and a configuration $\inet T\Delta$ we have that $\inet T\Delta\rw_{\mathcal R}^* \inet{U}{\e}$ if and only if
\[
\run {T\cdot \llbracket\Delta\rrbracket}{{\color{black}\llbracket{\mathcal R}\rrbracket}}\e {U\cdot \e_l\cdot \e_r}*\e~.
\]
\end{proposition}


\subsection{Petri Nets}

\newcommand\Mf{\mathcal M_f}

Petri nets are a type of discrete event dynamic system, first introduced in~\cite{petri1962kommunikation} as a graphical model of concurrency. They have since found real-world applications in a range of areas, including transport, manufacturing, fault diagnosis, and power systems~\cite{petriapplications}.

We define a Petri net as a pair $(P, T)$ of finite sets of \emph{places} $P$ and \emph{transitions} $T\subset \Mf(P)\times \Mf(P)$, where $\Mf$ denotes finite multisets. A \emph{state} $s\in\Mf(P)$ is a distribution of \emph{tokens} over the places $P$. A transition $t=(t^-,t^+)$ may \emph{fire} as follows, where $\subseteq$, $(\setminus)$ and $\uplus$ are multiset inclusion, difference, and union respectively.
\[
	s~\mapsto (s \setminus t^-)\uplus t^+ \qquad (t^-\subseteq s)
\]
Petri nets encode directly in the RMC with a signature $\Sigma=\{\circ\}$ consisting of only a constant $\circ$ for tokens, and the set of places $P$ as the locations. A state $s$ embeds as the memory $\embed s=S_P$ consisting of stacks of tokens, where each stack $S_p$ holds the number of tokens at $p$ in $s$. Writing $[p_1,\dots,p_n]$ for a multiset of places, transitions embed as below, and a Petri net as the term $\embed{(P,T)}=(\sum_{t\in T}\,\embed t)^\ast$. 
\[
\begin{array}{l}
	\embed{\,([p_1,\dots,p_n],[q_1,\dots,q_m])\,} ~=~ 
\\  \qquad\qquad\qquad\qquad \term{p_1<\circ>;..;p_n<\circ>;[\circ]q_1;..;[\circ]q_n}
\end{array}
\]

\begin{proposition} 
For a Petri net $(P,T)$ and states $s,s'\in\Mf(P)$, we have $s\mapsto^*s'$ if and only if
\[
	\run{\embed{s}}{`{\color{black}\embed{(P,T)}}}{\e}  {\embed{s'}}{\star}{\e}~.
\]
\end{proposition}



\section{Conclusion and Further Research}

We believe that by exposing the duality of relational programming in syntax we have arrived at a natural and convincing foundational model of the paradigm. 
We hope that the RMC will allow the further development of tools and reasoning techniques for relational programming,  similar to those founded in the $\lambda$-calculus,  which have proved so important for functional programming: type systems and denotational semantics,  operational semantics,  equational reasoning and confluent reduction.   
In this paper,  we proved fundamental results supporting this new calculus,  focussing on showing it meets the design criteria set out and justified in the introduction.  We now outline several areas of future work.  

\subsection{Further Research}

\emph{Foundations for functional logic programming.}
It appears straightforward to conceive of a higher-order RMC,  based on the Functional Machine Calculus (FMC) \cite{Power:kappa-cat,Heijltjes:FMCI,Heijltjes:FMCII},  which would have potential as a foundational model of \emph{functional} logic programming. 
Already,  the FMC,  which seeded this work,  has shown how to seamlessly integrate functional programming and global state,  by the use of locations --- which are also present in the RMC.  Thus,  there is potential,  even,  for a unification of logic,  functional \emph{and imperative} programming.

\emph{Bridging programming languages and string diagrams.}
String diagrams have found practical applications in an increasing range of domains
\cite{coecke_interacting_2011, Perdrix:completeness-ZX, Hasegawa-traced,DBLP:conf/rta/Mellies14}.  This often involves their implementation at scale: for example,  the ZX-diagrams used for quantum circuit optimization may have on the order of $10^4$--$10^7$ nodes.  Studies into rewriting theory for string diagrams have lead to representations which are ameenable to efficient manipulation by a computer \cite{Frobenius, Frobenius2},  including at scale. 
And,  as string diagrams grow,  tools to specify,  represent and reason about them become more important: in particular,  while diagrams are convenient for human use while they are very small,  they very quickly become incomprehensible at scale.  This has lead to the introduction of circuit description languages which allow the high-level specification of algorithms which can yet be compiled down to circuit (or diagram) level  \cite{Selinger:quipper, Lindenhovius:string-diagram}.  
The RMC gives a principled foundation to programming languages for string diagrams --- at least for the class given by Frobenius monoids.  A higher-order RMC would allow the factorization of large diagrams and their compact specification and representation.  

Although the class of string diagrams this paper deals with is limited,  we believe it is possible to extend the RMC in a way which gives a unifying account of --- and uniform syntax and operational semantics to --- a wide class of string diagrams via \emph{unification modulo theory}.  For example,  by considering values from the theory of commutative monoids,  and performing unification modulo this theory,  we can represent as terms the string diagrams of \emph{graphical resource algebra} \cite{graphical-resource-algebra};  by considering values from the theory of Abelian groups,  we can represent as terms the string diagrams of \emph{graphical linear algebra} \cite{bonchi_graphical_2019} and the \emph{phase-free ZX calculus} \cite{coecke_interacting_2011}.   In each of these cases we have a constant $\val 0$, a binary function symbol $(\val +)$,  and terms
\[	
	\begin{array}{c@{\qquad}c}
	\term{[0]}& \term{Exy. <x>;<y>;[x+y]} \\
	\term{<0>}& \term{Exy. <x+y>;[y];[x]} 
	\end{array}
\]
 providing a second commutative monoid structure and its dual,  with the expected equational theory induced by the operational semantics of unification modulo.  We hope further for links to (differential) linear logic and monoidal differential categories via the observation that the free commutative (co-)monoid (co-)monad models the exponential modality \cite{DBLP:journals/mscs/BluteCS06}.   


\emph{Weighted relations.}
Although there are a wide range of languages with a relational semantics, the range of languages taking a \emph{weighted} relational semantics \cite{laird_weighted_2013,  tsukada_linear-algebraic_2022} of some form is vast.  Broadly conceived,  
these include linear-algebraic \cite{Perdix:algebraic-calculi, DBLP:journals/mscs/Vaux09}, probabilistic and quantum $\lambda$-calculi  \cite{ gay_quantum_2009},  and classes of string diagrams modelling similar domains
\cite{coecke_interacting_2011, coecke_new_2013, Perdrix:completeness-ZX,Bonchi:graphical-conjunctive-queries,bonchi:tape-diagrams, zanasi:logic-programming,  bonchi_graphical_2019}.  
In fact,  the RMC generalizes easily to include \emph{weighted} terms and thus weighted non-determinism.  
Restricting this general non-determinism in order to give accounts of stochastic or unitary processes in probabilistic or quantum domains,  respectively,  is less immediate and will require more sophisticated type systems,  as in \cite{Sabry:symmetric-pattern-matching,Chardonnet:curry-howard-reversible}.  However,  Frobenius monoids have already been used to axiomatize the exact conditioning operation of probabilistic programming languages  (albeit in the finite dimensional case) \cite{di_lavore_evidential_2023},  and thus investigation into a \emph{probabilistic machine calculus} appears warranted \cite{Silva:prob-kat}.  Quantum processes also admit Frobenius monoids as a fundamental primitive\cite{coecke_interacting_2011, coecke_new_2013, Perdrix:completeness-ZX}; furthermore,  they include as essential aspects both non-determinism and reversiblity,  and thus investigation into a \emph{quantum machine calculus} appears warranted, too.  

%


\bibliographystyle{plain}
\bibliography{RMC}


\end{document}